%% file: kerr2B.tex
\documentclass[a4paper,10pt]{article}
\usepackage{graphicx}
\usepackage{amsmath}
\usepackage{amsfonts}
\usepackage{amssymb}
\usepackage{amsmath}
\usepackage{latexsym}
\usepackage{color}
\usepackage{ntheorem}
\usepackage{braket}
\usepackage[colorlinks=true]{hyperref}
\input{LeeNewcommand}

\DeclareFontFamily{OT1}{rsfs}{}
\DeclareFontShape{OT1}{rsfs}{m}{n}{ <-7> rsfs5 <7-10> rsfs7 <10-> rsfs10}{}
\DeclareMathAlphabet{\mycal}{OT1}{rsfs}{m}{n}

{\catcode `\@=11 \global\let\AddToReset=\@addtoreset}
\AddToReset{equation}{section}

{\catcode `\@=11 \global\let\AddToReset=\@addtoreset}
\AddToReset{figure}{section}

{\catcode `\@=11 \global\let\AddToReset=\@addtoreset}
\AddToReset{table}{section}

\begin{document}

\title{Algorithmic characterization results for the Kerr-NUT-(A)dS space-time. II. \\
KIDs for the Kerr-(A)(de Sitter) family%
\thanks{Preprint UWThPh-2016-22. }
}
\author{Tim-Torben Paetz%
\thanks{E-mail:  Tim-Torben.Paetz@univie.ac.at}  \vspace{0.5em}\\  \textit{Gravitational Physics, University of Vienna}  \\ \textit{Boltzmanngasse 5, 1090 Vienna, Austria }}

\maketitle

\vspace{-0.2em}

\begin{abstract}
We characterize Cauchy data sets leading to  vacuum space-times with vanishing Mars-Simon tensor.
 This approach provides
an algorithmic procedure to check whether a given initial  data set  $(\Sigma,h_{ij},K_{ij})$
 evolves into  a space-time which is locally isometric
to a member of the Kerr-(A)(dS) family.
\end{abstract}

%\noindent
%\hspace{2.1em} PACs number: %04.20.Ex, 04.20.Ha

\tableofcontents

\section{Introduction}

%The Cauchy data viewpoint: KIDs for the Kerr-family
%\tim{needs to be done}
%Killing initial data set candidate
In  part I \cite{kerr1} of this work we have provided an algorithm to check whether a given vacuum space-time with  cosmological
constant $\Lambda\in\mathbb{R}$ belongs to the Kerr-NUT-(A)dS family, or, more generally,
 admits a (possibly complex) Killing vector field such that the associated Mars-Simon tensor (MST)  vanishes.
An issue, which complements these kind of problems, is to derive analog  results from the point of view of an initial value  problem.
This will be the main object of this article.

It is generally agreed that a main source of understanding
of
dynamical black hole space-times will come from numerical simulations. These make often use of
a $3+1$-decomposition of space-time, whence it becomes relevant to gain a better understanding of the $3+1$ features of the Kerr family (cf.\ \cite{kroon0} and the references given therein).
A $3+1$-decomposition hides the symmetries of a space-time  unless  it is  adapted to them.
%This occurs for instance  in numerical applications.
Invariant characterization results  are therefore of particular relevance for such kind of problems.

One would like to know whether a given Cauchy data set generates
a development which is isometric to a portion of
a Kerr or a Kerr-NUT-(A)dS-space-time.
Such a result has been obtained in \cite{kroon} for the Kerr space-time based on a $3+1$-splitting of the MST
and its space-time characterization as given in \cite{mars,mars2}.
However, in view of an algorithmic characterization, it suffers from a similar drawback as  corresponding space-time characterization results obtained in \cite{mars,mars2,mars_senovilla},
namely the need to solve PDEs before the results can be applied:
Given vacuum Cauchy  data $(\Sigma,h_{ij},K_{ij})$, it is a  non-trivial issue
to check whether there exists a scalar function $\sigma$ and a vector field $Y^i$ complementing them to Killing initial data which generate a space-time with a Killing vector field.

As for the space-time approach we shall employ the restrictions, coming from the requirement that the emerging space-time admits a Killing vector field
w.r.t.\  which the MST vanishes, to show that there is (up to rescaling) at most one candidate tuple $(\sigma,Y^i)$.
That yields an algorithm to check whether a given set of  Cauchy data $(\Sigma,h_{ij},K_{ij})$, solution to the vacuum constraint equations,
emerges into a space-time which admits a KVF w.r.t.\ which the MST vanishes.
Moreover, it can be extended to
provide an algorithm to check  whether the triple  $(\Sigma,h_{ij},K_{ij})$ constitutes  Kerr-NUT-(A)dS data, by which we mean that the Cauchy data evolve into a vacuum space-time which is  locally isometric to a member of the Kerr-NUT-(A)dS family.
The procedure will be algorithmic in the sense that, given $(\Sigma,h_{ij},K_{ij})$, only differentiation and computation of roots is needed without any need to solve
differential equations.
So far such an algorithmic test has  been given for Schwarzschild data \cite{kroon0},
for Kerr-data \cite{lobo},  and  for Petrov type D-data \cite{lobo2}.
%\tim{add sth; compare results with Juan's: Ernst-potential not regarded as an unknown, iff, which Kerr-space-time is generated}

The paper is organized as follows: In Section~\ref{section_preliminaries} we will review definition and some properties of the MST, a space-time characterization
result for the Kerr-(A)dS family based on this tensor, as well as the notion of Killing initial data sets (KIDs).
The definition of the MST  comes along with a scalar function $Q$,  which can be defined in several different ways.
In Section~\ref{section_2}  we analyze  the vanishing of the MST on Cauchy surfaces for different choices of $Q$
and investigate the equivalence of these choices.
  In Section~\ref{section_3} we construct candidates for solving the KID equation and  characterize conditions under which these
candidates are in fact KIDs. This way we are led to an algorithmic characterization of Cauchy data which generate  $\Lambda$-vacuum space-times
which admit a Killing vector field whose associated MST vanishes, cf.\ Theorem~\ref{thm_first_main_result}. A somewhat shortened version of  Theorem~\ref{thm_first_main_result} reads

\begin{theorem}
\label{thm_first_main_result_intro}
Consider Cauchy data $(\Sigma,h_{ij}, K_{ij})$ which solve the vacuum constraint equations and satisfy
\begin{equation*}
 \mathrm{tr}(\mathcal{E}\cdot\mathcal{E}) \,\ne \,0
\;, \enspace
\mathrm{tr}(\mathcal{E}\cdot\mathcal{E})   -\frac{2}{3}\,\Lambda^2 \,\ne \, 0
\;, \enspace
\mathrm{tr}(\mathcal{E}\cdot\mathcal{E})   - \frac{1}{6}\,\Lambda^2 \,\ne \, 0
\;, \enspace
\mathrm{tr}(\mathcal{E}\cdot\mathcal{E})   - \frac{8}{3}\,\Lambda^2 \,\ne \, 0
\;,
\end{equation*}
where
\begin{equation*}
\mathcal{E}_{ij} \,:=\, \mathring R_{ij}  + K K_{ij} - K_{ik}K_j{}^k -\frac{2}{3}\Lambda h_{ij}
 -i\mathring\epsilon_{i}{}^{kl}\mcD_{k}K_{lj}
%\;, \quad \mathcal{E}^2\,:=\, \mathcal{E}_{ij}\mathcal{E}^{ij}
\;,
\end{equation*}
and where $ \mathring R_{ij}$ and $\mcD$ denote the Ricci tensor and the Levi-Civita covariant derivative of $h_{ij}$.

Then the emerging $\Lambda$-vacuum space-time admits a non-trivial (possibly complex) KVF such that the associated MST vanishes
(at least in some neighborhood of $\Sigma$)
if and only if  two certain scalars  (\eq{alg_cond2} and \eq{solvability_condB_norm}, cf.\ Section~\ref{section_3} for the details) which depend on $\mathcal{E}$, $\mcD\mathcal{E}$, $h$  and $K$ vanish.
\end{theorem}

Finally, this result is combined  in Section~\ref{section_4}  with well-known space-time characterizations of Kerr-(NUT-)(A)dS to end up with
an algorithmic characterization  of these space-times in terms of their Cauchy data, Theorem~\ref{thm_second_main_result}.

\section{Preliminaries}
\label{section_preliminaries}

In this section we fix the notation and recall some results which will be relevant for  the subsequent analysis.
We will be rather brief here, for more details we refer the reader to part I \cite{kerr1}.

\subsection{Mars-Simon tensor and the function $Q$}
\label{sect_MST}

Let $(\mcM, g)$ be a smooth $3+1$-dimensional space-time which admits a Killing vector field (KVF) $X$.
Let us denote by $C_{\mu\nu\sigma\rho}$ its \emph{conformal Weyl tensor}, while $F_{\mu\nu}:= \nabla_{\mu}X_{\nu} =  \nabla_{[\mu}X_{\nu]}$
denotes the \emph{Killing form}.
We  define the  \textit{Mars-Simon tensor (MST)}
(compare \cite{ IK} where somewhat different conventions are used)
as
\begin{equation}
\mathcal{S}_{\mu\nu\sigma\rho}\,:=\, \mathcal{C}_{\mu\nu\sigma\rho}  + Q\mathcal{Q}_{\mu\nu\sigma\rho}
\;,
\label{dfn_mars-simon}
\end{equation}
where
\begin{eqnarray}
\mathcal{Q}_{\mu\nu\sigma\rho}  &:=& -  \mathcal{F}_{\mu\nu}\mathcal{F}_{\sigma\rho} + \frac{1}{3}\mathcal{F}^2\mathcal{I}_{\mu\nu\sigma\rho}
\;,
\\
\mathcal{I}_{\mu\nu\sigma\rho} &:=& \frac{1}{4} (g_{\mu\sigma}g_{\nu\rho} -g_{\mu\rho}g_{\nu\sigma} + i\epsilon_{\mu\nu\sigma\rho} )
\;,
\\
\mathcal{F}^2 &:=& \mathcal{F}_{\mu\nu} \mathcal{F}^{\mu\nu} \;,
\end{eqnarray}
and where
\begin{eqnarray}
 \mathcal{C}_{\mu\nu\sigma\rho} &:=& C_{\mu\nu\sigma\rho} +i C^{\star}_{\mu\nu\sigma\rho}
\;,
\\
\mathcal{F}_{\mu\nu} &:=& F_{\mu\nu} +i F^{\star}_{\mu\nu}
\;,
\end{eqnarray}
denote the \emph{self-dual Weyl tensor} and the \emph{self dual Killing form}, respectively.\
% which satisfy
%%
%\begin{equation}
%\mathcal{F}_{\mu\nu} \,=\, i\mathcal{F}^{\star}_{\mu\nu}\;, \quad \mathcal{C}_{\mu\nu\sigma\rho} \,=\,
%i \mathcal{C}^{\star}_{\mu\nu\sigma\rho}
%\;.
%\label{self_dual}
%\end{equation}
%
At this stage $Q:\mcM\rightarrow \mathbb{C}$ is an arbitrary function on $\mcM$.
The MST is a \emph{Weyl field}, i.e.\ it has all the algebraic symmetries of the Weyl tensor.

Let us address the issue how the function $Q$ is to be chosen.
Denote by
\begin{equation}
\chi_{\mu} := 2X^{\alpha}\mathcal{F}_{\alpha\mu}
\end{equation}
 the \emph{Ernst 1-form}.
In a $\Lambda$-vacuum space-time  it is well-known to be closed. Thus, at least locally, there exists a scalar field $\chi$, the \emph{Ernst potential}, such that
$\chi_{\mu}=\nabla_{\mu}\chi$. Note that $\chi$ is only defined up to some additive complex ``$\chi$-constant''.
We further set
\begin{equation}
\mathcal{C}^2 \,:=\,   \mathcal{C}_{\mu\nu\sigma\rho} \mathcal{C}^{\mu\nu\sigma\rho}
\;.
\end{equation}
Supposing that the corresponding denominators are non-zero we  have the following natural%
\footnote{
``Natural'' in the sense that each of  these expressions is obtained by requiring a certain component of the MST to vanish, whence the function $Q$ necessarily needs to coincide with each of these definitions  whenever the MST vanishes, cf.\ Proposition~\ref{prop_Q} below.
}
 definitions for the function $Q$ \cite{mpss, kerr1}
\begin{eqnarray}
Q_0  &:=&  \frac{3}{2}\mathcal{F}^{-4} \mathcal{F}^{\mu\nu}\mathcal{F}^{\sigma\rho}\mathcal{C}_{\mu\nu\sigma\rho}
\,,
\label{def0_Q}
\\
Q_{\mathrm{ev}} &:=&  \frac{3\mathcal{F}^2  + 4\Lambda \chi \pm 3\sqrt{\mathcal{F}^2(\mathcal{F}^2 + 4\Lambda \chi )}}{\chi\mathcal{F}^{2}}
\;,
\label{defev_Q}
\\
Q_{\mathcal{F}} &:=&
\pm\sqrt{ \frac{3}{2}}\,\mathcal{F}^{-2} \sqrt{\mathcal{C}^2}
\;.
\label{yet_another_Q}
\\
 Q_{\mathcal{C}} &:=& \varkappa (\mathcal{C}^2)^{5/6}\Big( \pm\sqrt{\mathcal{C}^2}
-\sqrt{\frac{32}{3}}\,\Lambda\Big)^{-2}
\;, \quad \varkappa\in\mathbb{C}\setminus\{0\}
\;.
\label{yet_another_Q_2}
\end{eqnarray}
In our setting the square roots will be taken only of nowhere vanishing functions, so that one can
prescribe the choice of square root at one point
and extend
it by continuity to the whole manifold, and since no  branch
point  is ever met  the root will be smooth everywhere.
%The prescription of  complex  roots  is described  in \cite{mpss, kerr1}.
The definitions \eq{defev_Q} and \eq{yet_another_Q_2} both  involve a complex constant which is arbitrary at this stage.
The same is true for the choice of $\pm$.

In \cite{kerr1} we have established the following
%\tim{!!!!!!!!}
%
\begin{proposition}
\label{prop_Q}
Assume that the MST associated to some KVF $X$ vanishes for some function $Q$,
%that $\mathrm{Im}\Big(\frac{\sqrt{\mathcal{F}^2}}{Q\mathcal{F}^2-4\Lambda}\Big)$ has non-zero gradient somewhere,
and that  the inequalities
\begin{equation}
 \mathcal{C}^2 \,\ne \,0
\;, \quad
\mathcal{C}^2   - \frac{32}{3}\Lambda^2 \,\ne \, 0
\;,
\label{ineqs_C2_0}
\end{equation}
hold.
Then  \eq{def0_Q}, \eq{yet_another_Q} and \eq{yet_another_Q_2} are regular everywhere, and
there exists
a constant $\varkappa \in\mathbb{C}\setminus\{0\}$ and  a choice of $\pm$
such that
\begin{equation}
Q\,=\, Q_0 \,=\, Q_{\mathcal{F}} \,=\,  Q_{\mathcal{C}}
\,.
\end{equation}
Assume that, in addition,
\begin{equation}
\mathcal{C}^2   - \frac{8}{3}\,\Lambda^2 \,\ne \, 0\;,
\quad
\mathcal{C}^2 - \frac{128}{3}\,\Lambda^2 \,\ne \, 0
\label{ineqs_C2_1}
\end{equation}
 hold.
Then  \eq{defev_Q} is regular, as well, and
there exists an Ernst potential $\chi$, i.e.\ a choice of the $\chi$-constant,  such that
\begin{equation}
Q\,=\, Q_0 \,=\, Q_{\mathrm{ev}}\,=\, Q_{\mathcal{F}} \,=\,  Q_{\mathcal{C}}
\,.
\end{equation}
%Moreover, the function $Q$ satisfies
%\tim{needed?}
%%
%\begin{eqnarray*}
%%Q^2- 2Q\mathcal{F}^{-2}(4\Lambda +3\mathcal{F}^2 \chi^{-1}) +4\Lambda \mathcal{F}^{-4}(4\Lambda    - 3\mathcal{F}^{2}\chi^{-1}) &=& 0
%%\,,
%%\\
%\mathcal{F}^2\nabla_{\kappa}Q
%+ \frac{1}{4}\frac{Q\mathcal{F}^2   +20\Lambda }{Q\mathcal{F}^2 +2\Lambda} Q\nabla_{\kappa}\mathcal{F}^2 \,=\, 0
%\,.
%\end{eqnarray*}
\end{proposition}

\begin{remark}
\label{remark_weaker}
{\rm
The second condition in \eq{ineqs_C2_0} and the conditions in \eq{ineqs_C2_1} may be replaced by
\begin{equation}
\pm \sqrt{\mathcal{C}^2}   - \sqrt{\frac{32}{3}}\Lambda \,\ne \, 0
\;, \quad
\pm\sqrt{\mathcal{C}^2}   + \sqrt{\frac{8}{3}}\,\Lambda \,\ne \, 0\;,
\quad
\pm\sqrt{\mathcal{C}^2}   + \sqrt{\frac{128}{3}}\,\Lambda \,\ne \, 0
\;,
\end{equation}
respectively, and merely need to hold for one sign, depending on the sign which one needs to take in \eq{defev_Q}-\eq{yet_another_Q_2}
for  the MST to vanish.
}
\end{remark}

Proposition~\ref{prop_Q} shows that as long as one is interested in space-times with vanishing MST the definitions \eq{def0_Q}-\eq{yet_another_Q_2}
of the functions $Q$ are equally good in the sense that they are all necessary for a space-time to admit a KVF for which the associated MST vanishes.

In this article we want to analyze the vanishing of the MST in terms of an initial value problem.
To derive sufficient conditions which ensure the existence of a KVF w.r.t.\ which the MST vanishes one needs evolution equations for the MST.
More precisely, one would like to have  homogeneous equations  at hand which ensure that, given an appropriate set of zero initial data, the zero-solution is the only one.
While it does not seem to be possible to derive such equations for  $ Q_0$, $Q_{\mathcal{F}}$, and $Q_{\mathcal{C}}$, it can be done \cite{IK, mpss} for  $Q=Q_{\mathrm{ev}}$:
\begin{proposition}
\label{prop_evolution}
Consider  a smooth $3+1$-dimensional $\Lambda$-vacuum space-time which admits a KVF and which satisfies
(cf.\ Remark~\ref{remark_weaker})
%\tim{in terms of $\mathcal{C}^2$?... is it possible at all?}
%
\begin{equation}
 \mathcal{C}^2 \,\ne \,0
\;, \quad
\mathcal{C}^2   - \frac{32}{3}\Lambda^2 \,\ne \, 0
\;, \quad
\mathcal{C}^2  - \frac{8}{3}\,\Lambda^2 \,\ne \, 0
\;,\quad
\mathcal{C}^2 -\frac{128}{3}\Lambda^2 \,\ne \, 0
\;.
\end{equation}
Then the MST with $Q=Q_{\mathrm{ev}}$  satisfies   a regular  linear homogeneous  symmetric hyperbolic system of evolution equations,
\begin{eqnarray}
\nabla_{\beta} \mathcal{S}^{(\mathrm{ev})}_{\mu\nu\alpha}{}^{\beta}
&=&
 - Q_{\mathrm{ev}}\Big( \mathcal{F}_{\alpha\beta}\delta_{\mu}{}^{\gamma}\delta_{\nu}{}^{\delta}   -  \frac{2}{3}
  \mathcal{F}^{\gamma\delta} \mathcal{I}_{\alpha\beta\mu\nu}
\Big)X^{\lambda}\mathcal{S}^{(\mathrm{ev})}_{\gamma\delta\lambda}{}^{\beta}
\nonumber
\\
&&
\hspace{3em}
-4 \Lambda  \frac{  5  Q_{\mathrm{ev}}\mathcal{F}^2  +4\Lambda }{Q_{\mathrm{ev}}\mathcal{F}^2 + 8\Lambda}
 \mathcal{Q}_{\mu\nu\alpha\beta}
   \mathcal{F}^{-4}\mathcal{F}^{\gamma\delta} X^{\lambda}   \mathcal{S}^{(\mathrm{ev})}_{\gamma\delta\lambda}{}^{\beta}
\;.
\label{phys_ev}
\end{eqnarray}
\end{proposition}

Here and henceforth we denote the MST corresponding to the choice $Q=Q_{\mathrm{ev}}$ by $\mathcal{S}^{(\mathrm{ev})}_{\alpha\beta\mu\nu}$,
the one corresponding to  $Q=Q_{\mathcal{C}}$ by  $\mathcal{S}^{(\mathcal{C})}_{\alpha\beta\mu\nu}$, etc.
Because of Proposition~\ref{prop_Q} this distinction is less essential in the setting where the MST vanishes, whence we will occasionally omit the
superscript.

\subsection{Space-time characterization results}
\label{sec_results}

Our aim is to analyze the implications of the  space-time characterization results for the Kerr-NUT-A(dS) metrics (in particular for vanishing NUT-charge) obtained by Mars and Senovilla  \cite{mars_senovilla}
on a Cauchy  problem, and to derive an analog of \cite[Theorem~5.3]{kerr1}
for Cauchy data.
In paper I  \cite{kerr1} we have reviewed their  results.
Here, let us just recall a characterization result for the Kerr-A(dS) metric \cite{kerr1} which is a reformulation of a special case
of \cite[Theorem 1]{mars_senovilla},   and  which is the most important one for our purposes.
%A combination of the  results in  \cite{mars,mars2,mars_senovilla}, stated for convenience just for the Kerr case, gives (cf.\ \cite{kerr1}),

To this end  we define 4 real-valued functions
 $b_1$, $b_2$, $c$ and $k$ \cite{mars_senovilla}
(we assume that  $Q\mathcal{F}^2 - 4\Lambda $ is nowhere vanishing),
\begin{eqnarray}
 b_2-ib_1 &:=& - \frac{ 36 Q (\mathcal{F}^2)^{5/2} }{(Q\mathcal{F}^2-4\Lambda)^3}
\label{equation_b1b2}
\,,
\\
c &:=& - |X|^2-  \mathrm{Re}\Big(  \frac{6\mathcal{F}^2(Q\mathcal{F}^2+2\Lambda)}{(Q\mathcal{F}^2-4\Lambda)^2} \Big)
\,,
\label{equation_c}
\\
k  &:=&  \Big|\frac{36\mathcal{F}^2}{(Q\mathcal{F}^2-4\Lambda)^2}\Big| \nabla_{\mu}Z\nabla^{\mu}Z -  b_2Z +cZ^2 +\frac{\Lambda}{3} Z^4
\,,
\label{equation_k}
\end{eqnarray}
where
\begin{equation}
 Z \,:=\,  6\,\mathrm{Re} \Big( \frac{\sqrt{\mathcal{F}^2}}{Q\mathcal{F}^2-4\Lambda}\Big)
\,.
\label{equation_Z}
\end{equation}
Moreover, we set (because of  the assumptions \eq{second_condition1}-\eq{second_condition3} below the square roots will be  real)
\begin{eqnarray}
\text{for $\Lambda =0$:} &&  \zeta_1 \,:=\, \sqrt{ \frac{k}{c}}\;,
\label{first_condition1b}
\\
\text{for $\Lambda >0$:} &&  \zeta_1 \,:=\, \sqrt{-\frac{3}{\Lambda}\frac{c}{2} + \sqrt{\Big( \frac{3}{\Lambda} \frac{c}{2}\Big)^2  + \frac{3}{\Lambda} k}}
\label{first_condition2b}
\;,
\\
\text{for $\Lambda <0$:} &&  \zeta_1 \,:=\, \sqrt{-\frac{3}{\Lambda}\frac{c}{2} - \sqrt{\Big( \frac{3}{\Lambda} \frac{c}{2}\Big)^2  + \frac{3}{\Lambda} k}}
\label{first_condition3b}
\;.
\end{eqnarray}

\begin{theorem}[cf. \cite{mars_senovilla}]
\label{thm_charact}
% \ptcr{this is due to Mars and Senovilla, right? doesn't hurt to repeat this info here}
%\tim{if and only if}
Let $(\mcM,g)$ be a smooth $3+1$-dimensional $\Lambda$-vacuum space-time which admits a KVF $X$ such that  the associated MST vanishes for some function $Q$.
Assume  that $Q\mathcal{F}^2$ and $Q\mathcal{F}^2-4\Lambda$ are not identically zero, and that $\mathrm{Im}\Big(\frac{\sqrt{\mathcal{F}^2}}{Q\mathcal{F}^2-4\Lambda}\Big)$ has non-zero gradient somewhere.
%\tim{!!!!}
Then the functions $b_1$,$b_2$, $c$, and $k$ are constant.
Assume further that    $b_2=0$ and that
\begin{eqnarray}
\text{for $\Lambda =0$:} && c\, > \,0
\quad  (\Longrightarrow \quad k  \,\geq\, 0)
\;,
\label{second_condition1}
\\
\text{for $\Lambda >0$:} &&   c \,>\,  0 \quad  (\Longrightarrow \quad k\,\geq \, 0) \quad \text{or}
\nonumber
\\
&&      c\, \leq  \,0  \quad \text{and} \quad  k\,> \, 0\;,
\label{second_condition2}
\\
\text{for $\Lambda <0$:} &&  c\, > \,0   \quad \text{and} \quad k\,<\, \frac{3}{|\Lambda |} \frac{c^2}{4}
\quad  (\Longrightarrow \quad k\,\geq \, 0)
\;.
\label{second_condition3}
\end{eqnarray}
Then $(\mcM, g)$ is locally isometric to a Kerr-(A)dS space-time with parameters $(\Lambda,m,a)$, where
\begin{equation}
m\,=\,\frac{b_1}{2\big(\frac{\Lambda}{3} \zeta_1^2 + c\big)^{3/2}}\,, \quad a\,=\, \frac{\zeta_1}{\big(\frac{\Lambda}{3} \zeta_1^2 + c\big)^{1/2}}\,.
\end{equation}
The Schwarzschild-(A)dS limit is obtained for $a=0$, equivalently $k=0$.
\end{theorem}

\subsection{Killing initial data sets (KIDs)}

In order to define the MST  the emerging space-time needs to admit a KVF. This is ensured by so-called \emph{Killing initial data sets (KIDs)}.
%$(\Sigma,h_{ij},K_{ij},\sigma,Y^i)$:

%\subsubsection{Space-like KIDs}
\begin{theorem}[\cite{CBG,ChBActa,moncrief}, cf.\ \cite{beig}]
Consider the tuple $(\Sigma,h_{ij},K_{ij},\sigma,Y^i)$, where $(\Sigma,h_{ij})$ is a Riemannian 3-manifold, and where $K_{ij}$,  $\sigma$ and $Y^i$ are a  symmetric 2-tensor, a scalar function and a vector field on $\Sigma$, respectively.
Then there exists an (up to isometries) unique maximal globally hyperbolic space-time $(\mcM, g_{\mu\nu})$ such that
\begin{enumerate}
\item[(i)]
$(\mcM, g_{\mu\nu})$ solves Einstein's vacuum field equations $R_{\mu\nu}=\Lambda g_{\mu\nu}$,
\item[(ii)] $(\mcM, g_{\mu\nu})$  contains $\Sigma$ as a Cauchy surface with $h_{ij}$ and $K_{ij}$ being its first and second fundamental form, and
\item[(iii)] $(\mcM, g_{\mu\nu})$  admits a KVF $X^{\mu}$ with $(X^t,X^i)|_{\Sigma}=(\sigma,Y^i)$,
\end{enumerate}
 if and only if the vacuum constraint equations
\begin{eqnarray}
 \mathring R- |K|^2 + K^2  - 2\Lambda  &=& 0\;,
\label{constraint1}
\\
\mcD_j K_{i}{}^j - \mcD_i K &=& 0
\;,
\label{constraint2}
\end{eqnarray}
and the KID equations
\begin{eqnarray}
 \mcD_{(i}Y_{j)} + K_{ij}\sigma  &=& 0
\;,
\label{KID1}
\\
 \mcD_{i}\mcD_{j}\sigma  + \mcL_{Y}K_{ij}
-(\mathring R_{ij} + K K_{ij} -  2 K_i{}^kK_{jk} -\Lambda h_{ij})\sigma &=& 0
\;,
\label{KID2}
\end{eqnarray}
are fulfilled.
\end{theorem}

We use  $\,\mathring{}\,$ to denote objects associated to the Riemannian metric $h_{ij}$. $\mcD$~denotes the covariant derivative associated to $h_{ij}$.

\section{A vanishing MST arising from a space-like Cauchy problem}
\label{section_2}

Suppose  we have been given KIDs, i.e.\ a tuple $(\Sigma,h_{ij},K_{ij},\sigma,Y^i)$ which solves \eq{constraint1}-\eq{KID2}.
We want to extract conditions which characterize  the vanishing of the MST    on the initial surface $\Sigma$.

\subsection{Weyl tensor and Killing form in adapted coordinates}

%For reasons of convenience
To make computations as simple as possible  let us impose certain gauge conditions: We assume that the initial surface $\Sigma$ has (locally) been given in
\emph{adapted coordinates} $(t,x^i)$ with $\Sigma=\{t=0\}$, and we further assume a gauge where
%\tim{can one later argue that the formulas are independent of this choice?}
%
\begin{equation}
\ol g_{tt} =-1 \;, \quad \ol g_{ti} = 0\;, \quad \ol{\partial_t g_{t\mu}}=0\;,
\end{equation}
where here and henceforth an overbar means restriction of the corresponding space-time object to the initial surface $\Sigma$.
In such a  gauge the second   fundamental form reads
\begin{equation}
K_{ij}\,=\,\frac{1}{2}\ol{\partial_tg_{ij}}
\;,
\end{equation}
while the  traces of the connection coefficients on the initial surface become
\begin{equation}
\ol \Gamma^t_{tt} = \ol \Gamma^t_{ti}=\ol \Gamma^i_{tt} = 0\;, \quad
\ol \Gamma^t_{ij} = K_{ij}\;, \quad \ol \Gamma^i_{tj}=K_j{}^i\;, \quad \ol \Gamma^k_{ij} = \mathring \Gamma^k_{ij}
\;.
\end{equation}
The \emph{electric} and \emph{magnetic} part of the conformal Weyl tensor are given by
\begin{eqnarray}
 E_{ij} &:=& \ol C_{titj} \,=\, \mathring R_{ij}  + K K_{ij} - K_{ik}K_j{}^k -\frac{2}{3}\Lambda h_{ij}
\;,
\\
B_{ij} &:=& \ol C^{\star}_{titj}\,=\, - \mathring\epsilon_{i}{}^{kl}\mcD_{k}K_{lj}
\;.
\end{eqnarray}
(Note that the vacuum constraints imply
%that $B_{ij}$ is symmetric,
%
$
 \mathring\epsilon_{[i}{}^{kl}\mcD_{|k}K_{l|j]}= 0
$.)
It is useful to introduce the following fields on $\Sigma$,
%\tim{first 2 needed?}
%
\begin{eqnarray}
%f_i &:=& \ol F_{ti} \,=\, \mcD_i\sigma + K_{ij}Y^j
%\;,
%\\
\mathcal{P}_i &:=& \ol{\mathcal{F}}_{ti} \,=\, \mcD_i\sigma + K_{ij}Y^j+ \frac{i}{2} \mathring \epsilon_i{}^{jk}\mcD_jY_k
\;,
\label{dfn_calP}
\\
\mathcal{E}_{ij} &:=&  \ol{\mathcal{C}}_{titj} \,=\, E_{ij} + i B_{ij}
\;.
\label{dfn_cal_E}
\end{eqnarray}
Moreover, it is  convenient to introduce a notation for the curls of the KVF $Y$ and the co-vector field $\mathcal{P}$,
\begin{eqnarray}
Z^i &:=&  %\ol\epsilon_{t}{}^{ijk}\ol F_{jk} \,=\,
 \mathring \epsilon^{ijk}\mcD_jY_k
\;,
\\
\mathcal{Q}_i &:=&  \mathring \epsilon_i{}^{jk} \mcD_j\mathcal{P}_k
\;.
\label{dfn_calQ}
\end{eqnarray}

One straightforwardly  checks that it follows from the vacuum constraint equations \eq{constraint1}-\eq{constraint2} that $\mathcal{E}_{ij}$ satisfies the following relations,
\begin{eqnarray}
h^{ij}\mathcal{E}_{ij} \,=\,0\;, \quad
\mcD_j \mathcal{E}_{i}{}^{j} \,=\,
    i \mathring\epsilon_{i}{}^{jk}K_j{}^l \mathcal{E}_{kl}
\;.
\label{constraints_E}
\end{eqnarray}

For the sake of completeness let us collect some more useful formulas which are obtained by
employing the self-duality of the objects under consideration
%which will be useful later on
(to compute the transverse derivative of the self-dual  Weyl tensor  the contracted second
Bianchi identity has been used),
\begin{eqnarray}
\ol{\mathcal{F}}_{ij} &=&
% i \ol{\mathcal{F}}^{\star}_{ij} \,=\,
 - i \mathring \epsilon_{ijk} \mathcal{P}^{k}
\;,
\label{var_form_beg}
\\
\ol{\mathcal{C}}_{tijk} &=&
%i \ol{\mathcal{C}}^{\star}_{tijk} \,=\,
- i \mathring \epsilon_{jkl}\mathcal{E}_{i}{}^{l}
\;,
\qquad
\ol{\mathcal{C}}_{ijkl} \,=\, -\mathring\epsilon_{ij}{}^{p}\mathring \epsilon_{kl}{}^{q}\mathcal{E}_{pq}
\;,
\\
\ol{\nabla_t \mathcal{C}_{titj}} &=&
% %-\nabla_k\mathcal{C}_{tij}{}^k \,=\,
%i\mathring \epsilon_{(i}{}^{kl}\mcD_{|k}\mathcal{E}_{l|j)} + K_{(i}{}^k\mathcal{E}_{j)k}  -K\mathcal{E}_{ij} + \mathring\epsilon_{i}{}^{kp}\mathring\epsilon_{j}{}^{lq} K_{kl}\mathcal{E}_{pq}
%\\
%&=&
i\mathring \epsilon_{(i}{}^{kl}\mcD_{|k}\mathcal{E}_{l|j)} +3( K_{(i}{}^k\mathcal{E}_{j)k})\breve{}  - 2K\mathcal{E}_{ij}
\;,\phantom{xxx}
\label{trans_weyl}
\\
\ol{\nabla_t\mathcal{C}_{tijk}} &=&
% i\ol{\nabla_t\mathcal{C}^{\star}_{tijk}} \,=\,
 - i\mathring \epsilon_{jk}{}^{l}\ol{\nabla_t\mathcal{C}_{titl}}
\;,
\qquad
\ol{\nabla_t\mathcal{C}_{ijkl}} \,=\,
% i\ol{\nabla_t\mathcal{C}^{\star}_{ijkl}} \,=\,
 i\mathring\epsilon_{ij}{}^{p}\mathring\epsilon_{kl}{}^{q}\ol{\nabla_t\mathcal{C}_{tptq}}
\;,
%\\
%\ol{\nabla_t\mathcal{C}_{ijkl}} &=& - 2 \nabla_{[i}\mathcal{C}_{j]tkl} \,=\,  2i \epsilon_{kl}{}^{tp}\nabla_{[i}\mathcal{C}_{j]tpt}
%\\
% &=& -2i \mathring \epsilon_{kl}{}^{p}(\mcD_{[i}\mathcal{E}_{j]p} - K_{[i}{}^q\mathcal{C}_{j]q pt}
% - K_{[i }{}^q\mathcal{C}_{j]tpq})
%\\
% &=& -2i \mathring \epsilon_{kl}{}^{p}\mcD_{[i}\mathcal{E}_{j]p}
%+2 \mathring \epsilon_{kl}{}^{p} \mathring\epsilon_{qm[i}K_{j]}{}^q\mathcal{E}_{p}{}^m
% - 2 K_{k[i }\mathcal{E}_{j]l} +  2 K_{l[i }\mathcal{E}_{j]k}
%\\
%\ol{\nabla_t\mathcal{C}_{tijk}} &=&\frac{1}{2} i\epsilon_{ti}{}^{lp}\nabla_t\mathcal{C}_{lpjk}
%\\
%&=&  \mathring \epsilon_{i}{}^{lp} \mathring \epsilon_{kj}{}^{n}\mcD_{p}\mathcal{E}_{ln}
%- i \mathring \epsilon_{jk}{}^{n}K_{il}\mathcal{E}_{n}{}^{l}
%+  iK \mathring \epsilon_{jk}{}^{n}\mathcal{E}_{in}
%+2  i\mathring \epsilon_{i}{}^{lp}  K_{p[j }\mathcal{E}_{k]l}
%\\
%\ol{\nabla_t \mathcal{C}_{titj}} &=&  \frac{1}{2}i\epsilon_{tj}{}^{kl}\nabla_t \mathcal{C}_{tikl}
%\\
% &=& i  \mathring \epsilon_{i}{}^{kl}\mcD_{k}\mathcal{E}_{lj}
% + K_{i}{}^{k}\mathcal{E}_{jk}
%-  K \mathcal{E}_{ij}
% +\mathring \epsilon_{i}{}^{pq}    \mathring \epsilon_{j}{}^{kl}K_{pk }\mathcal{E}_{ql}
\label{var_form_end}
\end{eqnarray}
where $(.)\breve{}$ denotes the trace-free part of the corresponding 2-tensor w.r.t.\ $h_{ij}$.

\subsection{Vanishing of the MST  on $\Sigma$}

Let us compute the trace of the MST on the initial hypersurface $\Sigma$.
First of all we observe that
\begin{equation}
\mathcal{F}^2|_{\Sigma} \,=\, 4\ol{\mathcal{F}}_{ti}\ol{\mathcal{F}}^{ti} %+  \ol{\mathcal{F}}_{ij}\ol{\mathcal{F}}^{ij}
\,=\, - 4\mathcal{P}^2
\;.
\label{relation_F_P}
\end{equation}
We further obtain
\begin{equation}
\mathcal{Q}_{titj}|_{\Sigma} \,=\,  -\ol{\mathcal{F}}_{ti}\ol{\mathcal{F}}_{tj} + \frac{1}{3}\ol{\mathcal{F}}^2\ol{\mathcal{I}}_{titj}
\,=\, -(\mathcal{P}_i\mathcal{P}_j)\breve{}
\;.
\label{expr_F2}
\end{equation}
%
%where $(.)\breve{}$ denotes the trace-free part of the corresponding 2-tensor w.r.t.\ $h_{ij}$.
We conclude that
\begin{equation}
\mathcal{S}_{titj}|_{\Sigma}\,=\,\ol{\mathcal{C}}_{titj} + \ol Q\, \ol{\mathcal{Q}}_{titj}
\\
\,=\,  \mathcal{E}_{ij}- q (\mathcal{P}_i\mathcal{P}_j)\breve{}
\;,
\label{unproper_MST}
\end{equation}
where we have set
\begin{equation}
q\,:=\, Q|_{\Sigma}
\;.
\end{equation}
Note that $\mathcal{S}_{titj}|_{\Sigma}$ encompasses all independent components of the MST.
Consequently, the MST vanishes on $\Sigma$ for some function $q$ if and only if
\begin{equation}
 \mathcal{E}_{ij} \,=\, q (\mathcal{P}_i\mathcal{P}_j)\breve{}
\;.
\label{main_one}
\end{equation}

%It is convenient to set $Q=Q_0$. Assuming $\mathcal{F}^2\ne 0$ we then  find
%
%
%%
%\begin{eqnarray}
%\mathcal{F}^4\mathcal{S}^{(0)}_{titj}|_{\Sigma}&=& \ol{\mathcal{F}}^4\ol{\mathcal{C}}_{titj} + \ol Q_0\ol{\mathcal{F}}^4 \ol{\mathcal{Q}}_{titj}
%\\
%&=& 16\mathcal{P}^4 \mathcal{E}_{ij}- 24\mathcal{P}^k\mathcal{P}^l \mathcal{E}_{kl}  (\mathcal{P}_i\mathcal{P}_j)\breve{}
%\;.
%\end{eqnarray}
%

It follows immediately from the definition of the function $Q_0$ \cite{kerr1} that whenever the MST vanishes (in space-time, on some hypersurface, or merely at one point)
the function $Q$ needs to coincide with $Q_0$ there, supposing that $\mathcal{F}^2\ne 0$.
In particular, $q=Q_0|_{\Sigma}$ if  the  MST restricted to $\Sigma$ vanishes for some function $q$,
\begin{equation}
 q= Q_0|_{\Sigma}= \frac{3}{2}\ol{\mathcal{F}}^{-4}\ol{ \mathcal{F}}^{\mu\nu}\ol{\mathcal{F}}^{\sigma\rho}\ol{ \mathcal{C}}_{\mu\nu\sigma\rho}
=   24\ol{\mathcal{F}}^{-4} \ol{ \mathcal{F}}^{ti}\ol{\mathcal{F}}^{tj}\ol{ \mathcal{C}}_{titj}
= \frac{3}{2}\mathcal{P}^{-4}\mathcal{P}^i\mathcal{P}^j \mathcal{E}_{ij}
\;.
\label{expr_Q0}
\end{equation}
 We are  thus led to the following result:

\begin{lemma}
\label{some_intermediate}
Suppose that $\mathcal{P}^2\ne0$.
Then a Killing initial data set $(\Sigma,h_{ij},K_{ij},\sigma,Y^i)$  yields a $\Lambda$-vacuum space-time with a KVF such that the
associated MST vanishes \underline{on $\Sigma$} for some function $Q$
%Then the restriction  to $\Sigma$ of the MST associated to the Killing initial data set $(\Sigma,h_{ij},K_{ij},\sigma,Y^i)$ vanishes  for some function $Q$
 if and only if
\begin{equation}
\mathcal{E}_{ij} \,=\, \frac{3}{2}\mathcal{P}^{-4}\mathcal{P}^k\mathcal{P}^l \mathcal{E}_{kl}  (\mathcal{P}_i\mathcal{P}_j)\breve{}
\;,
\label{main_condition}
\end{equation}
and in that case $q=Q_0|_{\Sigma}$.
\end{lemma}

\begin{remark}
{\rm
Equation \eq{main_condition} is of the same form as the corresponding equation on $\scri$ derived in \cite{mpss}. Note, however, that the tensors
$\mathcal{E}_{ij}$ and $\mathcal{P}_i$ obey different equations as the corresponding ones in \cite{mpss}, whence $ \mathcal{P}^k\mathcal{P}^l \mathcal{E}_{kl}$
behaves differently, cf.\ the considerations below.
}
\end{remark}

\subsection{Equivalence of the choices of $Q$ on $\Sigma$}

In the previous section we made the choice $Q=Q_0$.
%It follows from the definition of the function $Q_0$ that whenever the MST vanishes (in space-time, on some hyper-surface, or merely at one point) the function $Q$ needs to coincide with $Q_0$ there, supposing that $\mathcal{F}^2\ne 0$.
Here, analog to the proceeding in \cite{kerr1}, our goal is to solve, on $\Sigma$, the equation $\mathcal{S}_{\alpha\beta\mu\nu}=0$ for the KVF,
and to do that it will be key to choose $Q=Q_{\mathcal{C}}$.
Moreover, in order to employ the evolution equations for the MST the choice $Q=Q_{\mathrm{ev}}$ is essential.

 Compared to \cite{kerr1}, though, we consider  a space-like hypersurface
rather than the full  space-time, whence Proposition~\ref{prop_Q} does not apply.
A priori it might e.g.\ happen that    $\mathcal{S}^{(\mathrm{ev})}_{\alpha\beta\mu\nu}$
does not vanish on some space-like hypersurface $\Sigma$  while  $\mathcal{S}^{(0)}_{\alpha\beta\mu\nu}$ or  $\mathcal{S}^{(\mathcal{C})}_{\alpha\beta\mu\nu}$  do. (On the other hand, if $\mathcal{S}^{(\mathrm{ev})}_{\alpha\beta\mu\nu}$ vanishes on $\Sigma$ it follows from the evolution equations
that it also vanishes off $\Sigma$, whence it follows from Proposition~\ref{prop_Q} that $\mathcal{S}^{(0)}_{\alpha\beta\mu\nu}$ and $\mathcal{S}^{(\mathcal{C})}_{\alpha\beta\mu\nu}$ needs to
vanish, as well.)
We therefore aim to make sure that whenever there exists a function $Q$ for which the MST vanishes  on a space-like hypersurface $\Sigma$,
then $Q=Q_0=Q_{\mathrm{ev}}=Q_{\mathcal{C}}$ holds there.
Along the way we will obtain some relations which will be crucial later on.

First of all we employ the KID equations \eq{KID1}-\eq{KID2} to compute
\begin{eqnarray}
\mcD_{i} \mathcal{P}_{j}
&\equiv&
\mcD_{i}\mcD_j\sigma + Y^k \mcD_{i} K_{jk}+ K_{j}{}^{k}\mcD_{(i}Y_{k)}  + K_{j}{}^{k}\mcD_{[i}Y_{k]}
+ \frac{i}{2}\mcD_{i}Z_j
\phantom{xxxxx}
\\
&=& \sigma \Big(\mathcal{E}_{ij} -\frac{\Lambda}{3}  h_{ij} \Big) + \mathring\epsilon_{ikl}B_{j}{}^{k}Y^l
-\frac{1}{2} \mathring\epsilon_{jkl}K_{i}{}^{k} Z^l
%+i\sigma B_{ij}
+ i\mathring \epsilon_{jk}{}^lK_{i}{}^{k}\mcD_l\sigma
\nonumber
\\
&&
 + \underbrace{ i\mathring \epsilon_{ijk} \mathring R_{l}{}^k Y^l + i\mathring \epsilon_{jkl}  \mathring R_{i}{}^{l} Y^k
-\frac{1}{2}i\mathring R\mathring \epsilon_{ijk}Y^k}_{=i\mathring \epsilon_{ikl}  \mathring R_{j}{}^{l} Y^k+ \frac{1}{2}i\mathring R\mathring \epsilon_{ijk}Y^k}
\\
&=& \sigma\Big(\mathcal{E}_{ij} -\frac{\Lambda}{3}  h_{ij} \Big) + i \mathring\epsilon_{ik}{}^{l}\Big(\mathcal{E}_{jl} -\frac{\Lambda}{3}h_{jl}\Big)Y^k+ i\mathring \epsilon_{jkl}K_{i}{}^{k}\mathcal{P}^l
\nonumber
\\
&&
+ \underbrace{ \frac{1}{2}i(|K|^2-K^2) \mathring\epsilon_{ijk}Y^k
+ i \mathring\epsilon_{[i}{}^{kl}( KK_{j]k} Y_l - K_{j]p}K_k{}^pY_l+  K_{j]k} K_{l}{}^{p}Y_p )}_{=0}
\nonumber
\\
&&+ \underbrace{i \mathring\epsilon_{(i}{}^{kl} (  KK_{j)k}Y_l - K_{j)p}K_k{}^pY_l - K_{j)k} K_{l}{}^{p}Y_p )}_{=0}
\;.
\label{deriv_P}
\end{eqnarray}
In particular,
\begin{equation}
\mcD_{i} \mathcal{P}^i
\,=\, -\Lambda \sigma
\;.
\label{div_P}
\end{equation}
From \eq{deriv_P} we compute
\begin{eqnarray}
 \mathcal{P}^2\mcD_i \mathcal{P}^2 &=&  2 \mathcal{P}^2\mathcal{P}^j\mcD_i \mathcal{P}_j
\\
 &=&  2 \mathcal{P}^2\Big(\sigma\mathcal{P}^j \mathcal{E}_{ij}
  -  i\mathring\epsilon_{ikl}\mathcal{P}^j \mathcal{E}_{j}{}^{k}Y^l
-\frac{\Lambda}{3}  \sigma \mathcal{P}_i
 +i\frac{\Lambda}{3}   \mathring\epsilon_{ijk}\mathcal{P}^jY^k
\Big)
\\
 &=& \Big( 2\mathcal{A}-\frac{2}{3}\Lambda   \mathcal{P}^2\Big)(\sigma \mathcal{P}_i  - i\mathring\epsilon_{ijk}\mathcal{P}^jY^k)
\;,
\label{deriv_P2_0}
\end{eqnarray}
where we have set
\begin{equation}
\mathcal{A} \,:=\, \mathcal{P}^k\mathcal{P}^l \mathcal{E}_{kl}
\;.
\label{dfn_cal_A}
\end{equation}
We would like to derive an algebraic relation between $\mathcal{A}$ and $\mathcal{P}^2$.
%By contraction  with $\mathcal{P}^i$ we observe that  \eq{main_condition}  implies
%\tim{needed???}
%%
%\begin{equation}
%\mathcal{P}^2 \mathcal{P}^j \mathcal{E}_{ij}-\mathcal{A} \mathcal{P}_i \,=\, 0
%\;.
%\label{contr_main}
%\end{equation}
%
We use \eq{constraints_E}, \eq{deriv_P}-\eq{div_P} and \eq{deriv_P2_0} to compute the divergence of  \eq{main_condition} (for $\mathcal{P}^2\ne 0$),
\begin{eqnarray}
0 &=& \mcD^j(2\mathcal{P}^4 \mathcal{E}_{ij} - 3\mathcal{A} (\mathcal{P}_i\mathcal{P}_j)\breve{})
\\
&\equiv&  4\mathcal{E}_{i}{}^{j}\mathcal{P}^2\mcD_j\mathcal{P}^2 + 2\mathcal{P}^4 \mcD^j\mathcal{E}_{ij}
- 3 \mathcal{P}_i\mathcal{P}^j\mcD_j\mathcal{A} - 3\mathcal{A}\mathcal{P}^j \mcD_j \mathcal{P}_i
\nonumber
\\
&&- 3\mathcal{A} \mathcal{P}_i\mcD_j\mathcal{P}^j
+\mathcal{P}^2\mcD_i \mathcal{A}+ \mathcal{A}\mcD_i\mathcal{P}^2
\\
&=&  9\sigma\mathcal{A}^2 \mathcal{P}^{-2}\mathcal{P}_i
+\frac{3}{2}i\mathcal{A}^2 \mathcal{P}^{-2} \mathring\epsilon_{ijk}\mathcal{P}^j Y^k
+  i\Lambda  \mathcal{A}\mathring\epsilon_{ijk}\mathcal{P}^jY^k
\nonumber
\\
&&
+\mathcal{P}^2\mcD_i \mathcal{A}
- 3 \mathcal{P}_i\mathcal{P}^j\mcD_j\mathcal{A}
- \mathcal{A}\mcD_i\mathcal{P}^2
\;.
\label{interm_A}
\end{eqnarray}
Contraction with $\mathcal{P}^i$ yields with \eq{deriv_P2_0}
% \eq{deriv_P} and \eq{contr_main}
%
\begin{equation}
2 \mathcal{P}^2\mathcal{P}^i\mcD_i \mathcal{A}
\,=\,  9\sigma\mathcal{A}^2
- \mathcal{A}\mathcal{P}^i\mcD_i\mathcal{P}^2
\,=\,  7\sigma\mathcal{A}^2
+\frac{2}{3}\Lambda  \sigma\mathcal{A}
 \mathcal{P}^2
\;.
\end{equation}
We insert this into \eq{interm_A} to deduce that
\begin{equation}
\mathcal{P}^4\mcD_i \mathcal{A}
- \mathcal{A} \mathcal{P}^2\mcD_i\mathcal{P}^2
\,=\, \mathcal{A}  \Big(\frac{3}{2}\mathcal{A}+\Lambda   \mathcal{P}^2\Big) (\sigma\mathcal{P}_i- i  \mathring\epsilon_{ijk}\mathcal{P}^j Y^k )
\;.
\label{interm_A2}
\end{equation}
Combining \eq{interm_A2} and \eq{deriv_P2_0} we end up with the  equation,
\begin{eqnarray}
\Big(\mathcal{A}-\frac{\Lambda}{3}   \mathcal{P}^2\Big) \mathcal{P}^2\mcD_i \mathcal{A}
&=&\Big(\frac{7}{4}\mathcal{A}+\frac{\Lambda}{6}   \mathcal{P}^2\Big) \mathcal{A} \mcD_i \mathcal{P}^2
\;.
\label{de_A}
\end{eqnarray}
For
\begin{equation}
\mathcal{P}^2\ne 0\;, \quad
\mathcal{A}\mathcal{P}^{-2}\ne 0\;, \quad \mathcal{A}  \mathcal{P}^{-2}+\frac{2}{3}\Lambda \ne 0
\;,
\label{new_ineqs}
\end{equation}
this can be written as
\begin{eqnarray}
 \mcD_i\log \mathcal{P}^2
&=&
\frac{4}{3}\frac{\mathcal{A}\mathcal{P}^{-2}-\frac{\Lambda}{3}   }{\mathcal{A}\mathcal{P}^{-2}+\frac{2}{3}\Lambda }
\mcD_i \log(\mathcal{A}\mathcal{P}^{-2})
\label{rln_A_P}
\\
&=&
\mcD_i \log \frac{(\mathcal{A}\mathcal{P}^{-2} + \frac{2}{3}\Lambda)^2}{(\mathcal{A}\mathcal{P}^{-2})^{2/3}}
\;.
\label{de_A2}
\end{eqnarray}
This PDE can be integrated straightforwardly to obtain a relation of the desired form,
\begin{equation}
\mathcal{P}^2
\,=\,
\mu \frac{(\mathcal{A}\mathcal{P}^{-2} + \frac{2}{3}\Lambda)^2}{(\mathcal{A}\mathcal{P}^{-2})^{2/3}}
\;,
\quad \mu\in\mathbb{C}\setminus\{0\}
\;.
\label{de_A3}
\end{equation}
This equation holds whenever $\mathcal{S}^{(0)}_{\alpha\beta\mu\nu}|_{\Sigma}=0$.

Let us return to the equivalence issue concerning the various choices of $Q$.
%We have \cite{kerr1}
%%
%\begin{eqnarray*}
%Q|_{\Sigma}=Q_{\mathrm{ev}} \quad \Longleftrightarrow \quad
%\chi|_{\Sigma} = 6\mathcal{F}^2 \frac{Q\mathcal{F}^2 + 2\Lambda}{(Q\mathcal{F}^2-4\Lambda)^2}
%%\\
%%\Longleftrightarrow \quad \nabla_i(Q\mathcal{F}^2)  =\frac{1}{2}(Q\mathcal{F}^2 - 4\Lambda) Q\nabla_i\chi
%%\\
%%\Longleftrightarrow \quad \nabla_i\chi =  \varkappa^{-1}\Big(\sqrt{\mathcal{C}^2}
%%-\sqrt{\frac{32}{3}}\,\Lambda\Big) (\mathcal{C}^2)^{-4/3}\nabla_i\mathcal{C}^2
%\end{eqnarray*}
%
It follows immediately from the computations  in \cite[Section~3.2]{kerr1}
 that
\begin{equation}
Q_{\mathcal{F}}|_{\Sigma} \,=\, Q_0|_{\Sigma}
\label{rel_QF_Q0}
\end{equation}
whenever $\mathcal{S}_{\alpha\beta\mu\nu}|_{\Sigma}=0$.
Note for this that  $Q=Q_{\mathcal{F}}$ can be derived algebraically from $\mathcal{S}_{\alpha\beta\mu\nu}=0$ without differentiation.

We employ \eq{de_A3} to express $\mathcal{F}^2|_{\Sigma}$ in terms of $Q\mathcal{F}^2|_{\Sigma}$, and finally in terms of $\mathcal{C}^2|_{\Sigma}$  (we set $\varkappa :=\mp3^{2/3}\sqrt{\frac{3}{2}} \,\mu^{-1}$),
\begin{eqnarray}
\mathcal{F}^2 |_{\Sigma}
&=&
\pm \Big(\frac{2}{3}\Big)^{1/6}\varkappa^{-1} \frac{( \ol Q_0\ol{\mathcal{F}}^2 - 4\Lambda )^2}{( \ol Q_0\ol{\mathcal{F}}^2)^{2/3}}
\label{F2_Q_relation}
%\,=\,  \Big(\frac{2}{3}\Big)^{1/6}\varkappa^{-1}\frac{( Q_F\mathcal{F}^2 - 4\Lambda )^2}{( Q_F\mathcal{F}^2)^{2/3}}
\\
&\overset{\eq{rel_QF_Q0}}{=}& \pm \sqrt{\frac{3}{2}}\,\varkappa^{-1}(\ol{\mathcal{C}}^2)^{-1/3}\Big( \pm\sqrt{\ol{\mathcal{C}}^2} -  \sqrt{ \frac{32}{3}}\Lambda \Big)^2
\;.
\label{F2_C2_relation}
\end{eqnarray}
Thus
(recall that in our current setting where $\mathcal{S}_{\alpha\beta\mu\nu}|_{\Sigma}=0$, there is no freedom to choose $\pm$, cf.\ Propostion~ \ref{prop_Q}),
\begin{equation}
Q_{\mathcal{F}}|_{\Sigma}
%\,\equiv\,\pm \sqrt{ \frac{3}{2}}\,\ol{\mathcal{F}}^{-2} \sqrt{\ol{\mathcal{C}}^2}
\,=\,
\varkappa (\ol{\mathcal{C}}^2)^{5/6}\Big( \pm\sqrt{\ol{\mathcal{C}}^2} -  \sqrt{ \frac{32}{3}}\Lambda \Big)^{-2}
\,=\, Q_{\mathcal{C}}|_{\Sigma}
\;.
\label{QF_QC_relation}
\end{equation}
By \eq{relation_F_P}, \eq{expr_Q0} and \eq{dfn_cal_A}
we have
\begin{equation}
\mathcal{A}\mathcal{P}^{-2}\,=\, - \frac{1}{6}\ol Q_0\ol{\mathcal{F}}^2
\;,
\label{relation_A_P_Q_F}
\end{equation}
whence   \eq{new_ineqs} is equivalent to
\begin{eqnarray}
&& Q \mathcal{F}^2|_{\Sigma}\ne 0\;, \quad Q\mathcal{F}^2|_{\Sigma} -4\Lambda \ne 0
\label{ineq_calC_1_0}
\\
\Longleftrightarrow &&  \mathcal{C}^2|_{\Sigma} \,\ne \,0
\;, \quad
\pm\sqrt{\mathcal{C}^2}|_{\Sigma}   - \sqrt{\frac{32}{3}}\Lambda \,\ne \, 0
\;,
\label{ineq_calC_1}
\end{eqnarray}
and  \eq{F2_Q_relation}-\eq{QF_QC_relation}  are well-defined.
Note that it depends on the sign in \eq{QF_QC_relation} for which the MST vanishes, for which sign \eq{ineq_calC_1} (and the corresponding conditions below) need to hold.

In view of $Q_{\mathrm{ev}}$ let us compute the restriction of the Ernst potential to $\Sigma$.
\begin{eqnarray}
\mcD_i\chi |_{\Sigma} \,=\,\ol \chi_i &=& 2\ol X^{\alpha}\ol{\mathcal{F}}_{\alpha i}
\\
&=& 2\sigma \mathcal{P}_{ i} - 2i\mathring\epsilon_{ijk}\mathcal{P}^jY^{k}
\\
&\overset{\eq{deriv_P2_0}}{=}& \Big(\mathcal{A}\mathcal{P}^{-2} -\frac{\Lambda}{3}   \Big)^{-1}\mcD_i \mathcal{P}^2
\label{ernst1}
\;,
\end{eqnarray}
supposing that, in addition to \eq{new_ineqs},
\begin{equation}
%\mathcal{P}^{2}  \,\ne\, 0\;, \quad
\mathcal{A}\mathcal{P}^{-2} -\frac{\Lambda}{3}  \,\ne\, 0
\;.
\label{new_ineq0}
\end{equation}
Because of \eq{relation_A_P_Q_F} this is equivalent to
\begin{equation}
Q\mathcal{F}^2|_{\Sigma} + 2\Lambda \,\ne\, 0 \quad \Longleftrightarrow \quad  \pm\sqrt{\mathcal{C}^2} |_{\Sigma}  + \sqrt{\frac{8}{3}}\,\Lambda \,\ne\,  0
\;.
\label{ineq_calC_2}
\end{equation}
It follows from  \eq{de_A3} and  \eq{ernst1} that  % (recall \eq{relation_F_P} and
\begin{eqnarray}
\mcD_i\chi |_{\Sigma}&=&  \mu  \Big(\mathcal{A}\mathcal{P}^{-2} -\frac{\Lambda}{3}   \Big)^{-1} \mcD_i\Big( \frac{(\mathcal{A}\mathcal{P}^{-2} + \frac{2}{3}\Lambda)^2}{(\mathcal{A}\mathcal{P}^{-2})^{2/3}}\Big)
\\
&=& \mcD_i \Big( 4\mu \frac{  \mathcal{A}\mathcal{P}^{-2}- \frac{\Lambda}{3} }{(\mathcal{A}\mathcal{P}^{-2})^{2/3}}\Big)
\\
&=& \mcD_i \Big( 4\mathcal{P}^2 \frac{  \mathcal{A}\mathcal{P}^{-2}- \frac{\Lambda}{3} }{(\mathcal{A}\mathcal{P}^{-2} + \frac{2}{3}\Lambda)^2}\Big)
%\\
%&=& \mcD_i \Big( 6 \ol{\mathcal{F}}^2 \frac{  \ol Q_0\ol{\mathcal{F}}^2 +2\Lambda  }{( \ol Q_0\ol{\mathcal{F}}^2- 4\Lambda)^2}\Big)
\;,
\end{eqnarray}
and thus, for an appropriate choice of the $\chi$-constant,
\begin{equation}
\chi |_{\Sigma} \,=\, 4\mathcal{P}^2 \frac{  \mathcal{A}\mathcal{P}^{-2}- \frac{\Lambda}{3} }{(\mathcal{A}\mathcal{P}^{-2} + \frac{2}{3}\Lambda)^2}
\overset{\eq{relation_A_P_Q_F}}{=}  6 \mathcal{F}^2 \frac{   Q_0\mathcal{F}^2 +2\Lambda  }{(  Q_0\mathcal{F}^2- 4\Lambda)^2}
\;.
\label{eqn_chi_A_P}
\end{equation}
Given $\chi|_{\Sigma}$ and $\mathcal{F}^2|_{\Sigma}$ this can be read as a \emph{quadratic} equation for $Q_0$, and it is precisely the same equation which is satisfied by $Q_{\mathrm{ev}}$ \cite{kerr1}.
Consequently,  one of its solutions satisfies
\begin{equation}
Q_{\mathrm{ev}}|_{\Sigma} \,=\, Q_0|_{\Sigma}
\end{equation}
for an appropriate choice of the  $\chi$-constant.
When working with $Q_{\mathrm{ev}}$ we also need to require $Q\mathcal{F}^2 + 8\Lambda \ne0$ \cite{kerr1}, or, equivalently,
\begin{equation}
\pm \sqrt{\mathcal{C}^2} |_{\Sigma}  + \sqrt{\frac{128}{3}}\,\Lambda \,\ne\,  0
\;.
\label{ineq_calC_3}
\end{equation}

%Using the self-duality of the complex Weyl tensor we compute
We have
\begin{equation}
\mathcal{C}^2|_{\Sigma}\,\equiv\, \mathcal{C}^{\alpha\beta\mu\nu}  \mathcal{C}_{\alpha\beta\mu\nu} |_{\Sigma}
\,=\,  16\,\mathcal{E}^2
\;,
\end{equation}
whence, \eq{ineq_calC_1}, and \eq{ineq_calC_2} \& \eq{ineq_calC_3}
are equivalent to
\begin{equation}
 \mathcal{E}^2|_{\Sigma} \,\ne \,0
\;, \quad
\pm\sqrt{\mathcal{E}^2}|_{\Sigma}   - \sqrt{\frac{2}{3}}\Lambda \,\ne \, 0
\;,
\label{all_ineqs_Cauchy1}
\end{equation}
and
\begin{equation}
\pm\sqrt{\mathcal{E}^2} |_{\Sigma}  + \sqrt{\frac{1}{6}}\,\Lambda \,\ne\,  0
\;,
\quad
\pm \sqrt{\mathcal{E}^2} |_{\Sigma}  + \sqrt{\frac{8}{3}}\,\Lambda \,\ne\,  0
\;.
\label{all_ineqs_Cauchy2}
\end{equation}
%
%respectively.

We have established the following lemma which is the Cauchy-surface-equivalent of  Proposition~\ref{prop_Q}:
\begin{lemma}
\label{lemma_equivalence_Qs}
Consider  a Killing initial data set, i.e.\ a tuple  $(\Sigma, h_{ij}, K_{ij}, \sigma, Y^i)$ which satisfies the vacuum constraints and the KID equations.
Assume that the restriction to $\Sigma$  of  the MST associated to the KVF $X$ generated by $(\sigma,Y^i)$ vanishes  for some function $Q$.
Assume further that the conditions
\eq{all_ineqs_Cauchy1} hold.
\begin{enumerate}
\item[(i)]
% the following inequalities hold on $\Sigma$,
%%
%\begin{equation}
% \mathcal{C}^2 \,\ne \,0
%\;, \quad
%\sqrt{\mathcal{C}^2}   - \sqrt{\frac{32}{3}}\Lambda \,\ne \, 0
%\;.
%\label{ineqs_C2_0B}
%\end{equation}
%
Then  $Q_0$, $Q_{\mathcal{F}}$ and $Q_{\mathcal{C}}$ are regular near $\Sigma$, and
there exists
a constant $\varkappa \in\mathbb{C}\setminus\{0\}$ and a choice of $\pm$
such that
\begin{equation}
Q|_{\Sigma}\,=\, Q_0 |_{\Sigma}\,=\, Q_{\mathcal{F}}|_{\Sigma} \,=\,  Q_{\mathcal{C}}|_{\Sigma}
\,.
\end{equation}
\item[(ii)]
Assume that, in addition, the inequalities \eq{all_ineqs_Cauchy2}  are valid.
%
%\begin{equation}
%\sqrt{\mathcal{C}^2}   + \sqrt{\frac{8}{3}}\,\Lambda \,\ne \, 0
%\label{ineqs_C2_1B}
%\end{equation}
%%
% holds.
Then, the function $Q_{\mathrm{ev}}$ is   regular near $\Sigma$ as well, and
there exists an Ernst potential $\chi$ for $ Q_{\mathrm{ev}}$ such that
\begin{equation}
Q|_{\Sigma}\,=\, Q_0|_{\Sigma} \,=\, Q_{\mathrm{ev}}|_{\Sigma}\,=\, Q_{\mathcal{F}} |_{\Sigma}\,=\,  Q_{\mathcal{C}}|_{\Sigma}
\,.
\end{equation}
\end{enumerate}
\end{lemma}

Whenever  \eq{all_ineqs_Cauchy1}-\eq{all_ineqs_Cauchy2} hold,
 the evolution equations \eq{phys_ev} for $\mathcal{S}^{(\mathrm{ev})}_{\alpha\beta\mu\nu}$ are regular, at least sufficiently close to $\Sigma$,
and imply by standard result for symmetric hyperbolic systems
 that the MST vanishes in some neighborhood of $\Sigma$:
%\tim{compare with Juan's result}

\begin{corollary}
\label{cor_vanishing_MST}
Let  $(\Sigma, h_{ij}, K_{ij}, \sigma, Y^i)$  be vacuum KIDs which satisfy \eq{all_ineqs_Cauchy1}-\eq{all_ineqs_Cauchy2}.
The MST of the emerging  space-time  $(\mcM,g_{\mu\nu})$ associated to the KVF generated by $(\sigma, Y^i)$  vanishes in some neighborhood of $\Sigma$ for some function $Q$  if and only if
$\mathcal{S}^{(\mathcal{C})}_{\alpha\beta\mu\nu}|_{\Sigma}=0$ (or
$\mathcal{S}^{(0)}_{\alpha\beta\mu\nu}|_{\Sigma}=0$ or
$\mathcal{S}^{(\mathcal{F})}_{\alpha\beta\mu\nu}|_{\Sigma}=0$ or
$\mathcal{S}^{(\mathrm{ev})}_{\alpha\beta\mu\nu}|_{\Sigma}=0$
).
\end{corollary}

The main advantage of the equation  $\mathcal{S}^{(\mathcal{C})}_{\alpha\beta\mu\nu}|_{\Sigma}=0$, as e.g.\ opposed to $\mathcal{S}^{(0)}_{\alpha\beta\mu\nu}|_{\Sigma}=0$, cf.\  \eq{main_condition}, is that it can be solved for $\mathcal{P}_i$ (supposing that a solution exists after all).
% at all.
%This will be of great relevance later on.

\section{Construction of solutions to the KID\,equations}
\label{section_3}

\subsection{Candidates for solving the KID equations}

Similar to the proceeding in \cite{kerr1} we do not want to assume that Killing initial data  $(\Sigma,h_{ij},K_{ij},\sigma,Y^i)$  have been given but only Cauchy data $(\Sigma, h_{ij}, K_{ij})$.
Indeed, it is  a grievance of  Corollary~\ref{cor_vanishing_MST}  that, given  $(\Sigma,h_{ij},K_{ij}$), it is a  non-trivial and non-algorithmic  issue
to check whether there exists $(\sigma,Y^i)$ complementing them to Killing initial data. Only then, it is straightforward to check whether $\mathcal{S}_{\mu\nu\sigma\rho}|_{\Sigma}=0$
holds.
%Also, in the case where one is able to check that  it does not hold, there might exist another KVF for which  \eq{main_condition2}  does hold.
We therefore intend to derive conditions from the equation $\mathcal{S}_{\mu\nu\sigma\rho}|_{\Sigma}=0$ which impose restrictions on  $(\sigma,Y^i)$.
As in the space-time case \cite{kerr1} it turns out that
% these conditions are so restrictive that
 up to rescaling  only one candidate  remains.

\subsubsection{Candidate fields}

Let us assume for the time being that we have been given a $\Lambda$-vacuum space-time which admits a KVF for which the associated MST vanishes, and which moreover satisfies $\mathcal{F}^2|_{\Sigma}\ne0$ and $Q\mathcal{F}^2|_{\Sigma} +2\Lambda\ne0$, or, equivalently,
\begin{equation}
\mathcal{P}^2\,\ne \, 0\;,\quad    q\mathcal{P}^2 - \frac{\Lambda}{2}\,\ne \, 0
\;.
\label{spec_ineqs}
\end{equation}
We will collect a  number of necessary conditions, which need to be satisfied in any such space-time. Among other things,
they will provide candidates for $\sigma$ and $Y^i$.

It has been shown in \cite{kerr1} that in vacuum space-times with vanishing MST and in which the space-time analog of \eq{spec_ineqs}  holds,
the self-dual Killing form $\mathcal{F}_{\mu\nu}$ and the function $Q$ necessarily satisfy the equation%
%following set of equations
%
\footnote{
The equation can be integrated to express $\mathcal{F}^2$ in terms of $Q\mathcal{F}^2$. The restriction of this  solution to $\Sigma$ recovers \eq{F2_Q_relation}.
}
\begin{equation}
%\mathcal{S}_{\alpha\beta\mu\nu} \,\equiv\, \mathcal{C}_{\alpha\beta\mu\nu}- Q\Big(\mathcal{F}_{\alpha\beta}\mathcal{F}_{\mu\nu} - \frac{1}{3}\mathcal{F}^2 \mathcal{I}_{\alpha\beta\mu\nu}\Big) &=& 0
%\;,
%\label{MST_main_cond}
%\\
\nabla_{\mu} Q + \frac{1}{4} \frac{Q\mathcal{F}^2 + 20 \Lambda}{Q\mathcal{F}^2 + 2\Lambda} Q\nabla_{\mu}\log\mathcal{F}^2 \,=\, 0
\;.
\label{PDE_Q}
\end{equation}
For vanishing MST the different choices for $Q$ are equivalent, whence we have    written $Q$ without any  subscript.
Let us consider its restriction to the initial surface $\Sigma$, where it suffices to take the spatial components into account, and complement it by
\eq{main_one}
\begin{eqnarray}
%\mathcal{S}_{titj}|_{\Sigma} \,\equiv\,
 \mathcal{E}_{ij} - q (\mathcal{P}_i\mathcal{P}_j)\breve{} &=& 0
\;,
\label{gen_vanishing_MST}
\\
 \mcD_i q + \frac{1}{4} \frac{q\mathcal{P}^2 -5 \Lambda}{q\mathcal{P}^2 - \frac{\Lambda}{2}} q \mcD_i \log\mathcal{P}^2  &=& 0
\;.
\label{q_PDE}
\end{eqnarray}
(Alternatively, \eq{q_PDE} may be obtained %from the computations made in this work
by  combining
 \eq{expr_Q0} and \eq{rln_A_P}.)
%In fact, we have the following
For the computations below, \eq{q_PDE} will often be needed in the form
\begin{equation}
 \mcD_i (q \mathcal{P}^2)\,=\, \frac{3}{4} \frac{q\mathcal{P}^2 + \Lambda}{q\mathcal{P}^2 - \frac{\Lambda}{2}} q \mcD_i \mathcal{P}^2
\;.
\label{PDE_QB}
\end{equation}
\begin{lemma}
\label{lemma_dq_qC}
Assume that \eq{spec_ineqs} and \eq{gen_vanishing_MST}  hold. Then \eq{q_PDE}  holds if and only if $q=Q_{\mathcal{C}}|_{\Sigma}$.
\end{lemma}
\begin{proof}
Equation  \eq{q_PDE}   is equivalent to
%
%\begin{eqnarray}
%&&\frac{4}{3}\frac{q\mathcal{P}^2 - \frac{\Lambda}{2}} {q\mathcal{P}^2 + \Lambda}  \mcD_i \log(q \mathcal{P}^2) \,=\,   \mcD_i \log\mathcal{P}^2
%\\
%&&\Longleftrightarrow \quad
%(q\mathcal{P}^2)^{-2/3} (q\mathcal{P}^2  + \Lambda)^2 \,=\,\varkappa \mathcal{P}^2
%\label{int_q_eqn}
%%\\
%%&&\Longleftrightarrow \quad  q\,=\,  \mathcal{Q}_{\mathcal{C}}|_{\Sigma}
%\;.
%\end{eqnarray}
\begin{equation}
\mcD_i\log\Big((q\mathcal{P}^2)^{-2/3} (q\mathcal{P}^2  + \Lambda)^2 \Big)\,=\, \mcD_i\log \mathcal{P}^2
\label{int_q_eqn}
\;.
\end{equation}
It follows from \eq{gen_vanishing_MST} that
$
q\mathcal{P}^2= \mp \sqrt{\frac{3}{32}}\sqrt{\mathcal{C}^2}
$,
and, after integration,  we deduce that \eq{int_q_eqn} is equivalent to
\begin{equation}
q  \,=\,\varkappa (\mathcal{C}^2)^{5/6}\Big(\pm \sqrt{\mathcal{C}^2}
-\sqrt{\frac{32}{3}}\,\Lambda\Big)^{-2}
\,=\, Q_{\mathcal{C}}|_{\Sigma}
\;,
\label{another_relation_q}
\end{equation}
as claimed.
\qed
\end{proof}

Our aim is  to derive candidate fields for $\sigma$ and $Y^i$  in terms of $\mathcal{P}_i$ and $q$.
Recall \eq{deriv_P}, whose derivation  required the KID equations. With \eq{gen_vanishing_MST} it becomes
\begin{equation}
\mcD_i \mathcal{P}_j
%&=& \sigma\Big( \mathcal{E}_{ij}-\frac{\Lambda}{3}h_{ij}\Big)  +i \mathring\epsilon_{ik}{}^{l}\Big(\mathcal{E}_{jl}-\frac{\Lambda}{3} h_{jl}\Big)Y^k+ i\mathring \epsilon_{jkl}K_{i}{}^{k}\mathcal{P}^l
%\\
\,=\, \sigma\Big( q \mathcal{P}_i\mathcal{P}_j -\frac{1}{3}(q\mathcal{P}^2 + \Lambda) h_{ij}\Big)  +i \mathring\epsilon_{ik}{}^{l}\Big(q \mathcal{P}_j\mathcal{P}_l -\frac{1}{3}(q\mathcal{P}^2 + \Lambda) h_{jl}\Big)Y^k
%\nonumber
+ i\mathring \epsilon_{jkl}K_{i}{}^{k}\mathcal{P}^l
\;.
\label{rep_deriv_P}
\end{equation}
Contraction with $\mathcal{P}^j$ yields
\begin{equation}
\mcD_i \mathcal{P}^2
\,=\, \frac{4}{3}\sigma\Big(  q \mathcal{P}^2-\frac{\Lambda}{2} \Big)\mathcal{P}_i
 + \frac{4}{3} i \mathring\epsilon_{i}{}^{kl}\Big(q \mathcal{P}^2-\frac{\Lambda}{2} \Big) Y_k\mathcal{P}_l
\;.
\label{deriv_P2}
\end{equation}
Contraction with $\mathcal{P}^i$ provides an expression for $\sigma$,
% (note that we assume that \eq{spec_ineqs} holds),
%
%\begin{equation}
%\sigma
%\,=\,
%\frac{3}{4}\Big(  q \mathcal{P}^2-\frac{\Lambda}{2} \Big)^{-1}\mathcal{P}^i\mcD_i \log\mathcal{P}^2
%\;.
%\end{equation}
%%
%Contraction with $\mathcal{E}^{ij}$ yields with \eq{main_condition2_0} and $\mathcal{E}^2\ne 0$
%\tim{and constraint}
%%
%\begin{equation}
%\sigma   \,=\, \mathcal{E}^{-2}\mathcal{E}^{ij}\mcD_i \mathcal{P}_j
%\;.
%\end{equation}
%
%Using \eq{main_condition2_0}, \eq{P_E_relation},   \eq{div_P} and \eq{relation_A_E_P} this becomes
%
%\begin{equation}
% \sigma   \,=\,
%- \sqrt{\frac{3}{8}} \Big(\sqrt{\mathcal{E}^2} + \frac{\Lambda}{\sqrt{6}}\Big)^{-1}\mathcal{P}^i \mcD_i\log \mathcal{P}^2
%\;.
%\end{equation}
%%
%With $q=- \sqrt{\frac{3}{2}} \mathcal{P}^{-2} \sqrt{\mathcal{E}^2}$ we end up with
%
\begin{equation}
 \sigma   \,=\,
 \frac{3}{4} \Big(  q\mathcal{P}^2 - \frac{\Lambda}{2}\Big)^{-1}\mathcal{P}^i \mcD_i\log \mathcal{P}^2
\;.
\label{cand_sigma}
\end{equation}
An application of   $\mathring\epsilon_{pq}{}^i$  to \eq{deriv_P2} and   relabeling indices gives
\begin{equation}
  Y_{[k}\mathcal{P}_{l]}
\,=\, \frac{i}{2}\sigma\mathring\epsilon_{kl}{}^m\mathcal{P}_m
-\frac{3}{8}i \Big(  q \mathcal{P}^2-\frac{\Lambda}{2} \Big)^{-1}\mathring\epsilon_{kl}{}^m\mcD_m \mathcal{P}^2
\;.
\label{anti_P_Y}
\end{equation}
%
%Contraction with $\mathcal{P}^l$ yields
%%
%\begin{equation}
% Y_{k}\mathcal{P}^2
%-   Y_{l}\mathcal{P}^l\mathcal{P}_{k}
%\,=\,
%-\frac{3}{4}i \Big(  q \mathcal{P}^2-\frac{\Lambda}{2} \Big)^{-1}\mathring\epsilon_{k}{}^{lm}\mathcal{P}_l \mcD_m \mathcal{P}^2
%\;.
%\label{equation1}
%\end{equation}
%%
%The information about the projection of $Y$ on $\mathcal{P}$ is not contained in \eq{deriv_P2}.
%To obtain a relation, we apply $\mathring\epsilon^{ijk}\mathcal{P}_k$ to \eq{rep_deriv_P},
%%
%\begin{eqnarray}
%\mathcal{P}_k Y^k\,=\,
%-\frac{3}{2} (q\mathcal{P}^2 + \Lambda)^{-1} \Big(i\mathring\epsilon^{ijk}\mathcal{P}_k\mcD_i \mathcal{P}_j
%+   K^{kl}\mathcal{P}_k\mathcal{P}_{l} -  K \mathcal{P}^2
%\Big)
%\;.
%\label{equation2}
%\end{eqnarray}
%%
%The equations \eq{equation1} and \eq{equation2} can be combined to obtain an expression for $Y$,
%%
%\begin{equation}
%\boxed{
% Y_{i}
%\,=\,
%-\frac{3}{4}i \Big(  q \mathcal{P}^2-\frac{\Lambda}{2} \Big)^{-1}\mathring\epsilon_{i}{}^{kl}\mathcal{P}_k\mcD_l\log \mathcal{P}^2
%-\frac{3}{2}\mathcal{P}^{-2}  (q\mathcal{P}^2 + \Lambda)^{-1} \Big(i\mathring\epsilon^{klm}\mathcal{P}_k\mcD_l \mathcal{P}_m
%+   K^{kl}\mathcal{P}_k\mathcal{P}_{l} -  K \mathcal{P}^2
%\Big)
%\mathcal{P}_{i}
%\;.
%}
%\end{equation}
%
%On the other hand, i
%
We insert  \eq{anti_P_Y} into \eq{rep_deriv_P} to  obtain the useful relation
\begin{equation}
\mcD_i \mathcal{P}_j
\,=\, \frac{1}{3}(q\mathcal{P}^2 + \Lambda)\Big(   i \mathring\epsilon_{ijk}Y^k  -\sigma h_{ij}   \Big)
+ \frac{3}{4}
\frac{ q\mathcal{P}^2}{ q\mathcal{P}^2-\frac{\Lambda}{2} }\mathcal{P}_j\mcD_i \log\mathcal{P}^2
+ i\mathring \epsilon_{jkl}K_{i}{}^{k}\mathcal{P}^l
\;.
\label{rep_deriv_P2}
\end{equation}
Supposing that
\begin{equation}
q\mathcal{P}^2 + \Lambda \,\ne\,0
\;,
\label{spec_ineqs2}
\end{equation}
its anti-symmetric part yields  an equation which can be solved for $Y$,
\begin{equation}
Y_i\,=\,
-\frac{3}{2} \Big( q\mathcal{P}^{2} + \Lambda\Big)^{-1}\Big[\frac{3}{4} i\mathring\epsilon_i{}^{kl} \frac{ q\mathcal{P}^{2}}{  q\mathcal{P}^{2} - \frac{\Lambda}{2}}\mathcal{P}_k\mcD_l\log\mathcal{P}^2
+ i \mathcal{Q}_i
 +K_{i}{}^{k}\mathcal{P}_{k}-  K \mathcal{P}_{i}
\Big]
\;.
\label{candidate_field_Y}
\end{equation}
Whenever a $\Lambda$-vacuum space-time  with %\eq{spec_ineqs_comb}
\eq{spec_ineqs} and  \eq{spec_ineqs2}  admits a KVF such that the associated MST vanishes, the corresponding KIDs
\emph{necessarily}  need to satisfy \eq{cand_sigma}
and \eq{candidate_field_Y} where $(q,\mathcal{P}_i)$ solve \eq{gen_vanishing_MST}-\eq{q_PDE} (cf.\ Proposition~\ref{interm_prop_1} below).

We would like to gain some  insight under which conditions the candidates  \eq{cand_sigma}
and \eq{candidate_field_Y} for $\sigma$ and $Y^i$ do  provide a solution of the KID equations.
For this purpose in turns out to be fruitful to derive a couple of relations between  the co-vector field $\mathcal{P}$   and the  function~$q$.

Of course, in general, there is no reason why $(\sigma,Y^i)$ should be real. As in \cite{kerr1} this does not cause any problems,
and we can enlarge our space-times of interest to those which admit a \emph{complex} KVF whose associated MST vanishes.

\subsubsection{Necessary conditions on $\mathcal{P}$}

Let us  compute the symmetric part of \eq{rep_deriv_P2} which provides a useful relation satisfied by $\mathcal{P}$ which does not involve $Y$,
\begin{equation}
\mcD_{(i} \mathcal{P}_{j)}
%&=& \sigma\Big( q \mathcal{P}_i\mathcal{P}_j -\frac{1}{3}(q\mathcal{P}^2 + \Lambda) h_{ij}\Big)  +i \mathring\epsilon_{(i}{}^{kl}\Big(q \mathcal{P}_{j)}Y_{[k}\mathcal{P}_{l]}
%+K_{j)k}\mathcal{P}_l
%\Big)
%\\
%&=&
\,=\,   -\frac{1}{3}\sigma(q\mathcal{P}^2 + \Lambda) h_{ij}
+\frac{3}{4}  \frac{q\mathcal{P}^2}{  q \mathcal{P}^2-\frac{\Lambda}{2}} \mathcal{P}_{(i}\mcD_{j)} \log\mathcal{P}^2
+ i \mathring\epsilon_{(i}{}^{kl} K_{j)k}\mathcal{P}_l
%\phantom{xxx}
\;.
\label{Equation2}
\end{equation}
Two special components will be of particular importance:
Its contraction  with~$\mathcal{P}^j$
\begin{equation}
\mathcal{P}^j\mcD_{(i} \mathcal{P}_{j)}
\,=\,   \frac{1}{6}\sigma(q\mathcal{P}^2 -2 \Lambda)\mathcal{P}_i
+\frac{3}{8}  \frac{q\mathcal{P}^2}{  q \mathcal{P}^2-\frac{\Lambda}{2}}\mcD_{i} \mathcal{P}^2
+ \frac{i}{2} \mathring\epsilon_{i}{}^{jk} K_{j}{}^l\mathcal{P}_k\mathcal{P}_l
\;,
\label{Equation2_contr}
\end{equation}
and  its trace
\begin{equation}
\mcD_{i} \mathcal{P}^i
\,=\,   -\Lambda\sigma
\;.
\label{divP}
\end{equation}
%
%The function $\sigma$ could be eliminated in the above equations using
%%
%\begin{equation}
%\sigma \,=\,  \frac{3}{2} \Big(  q\mathcal{P}^2 - \frac{\Lambda}{2}\Big)^{-1} \mathcal{P}^{-2}\mathcal{P}^i \mathcal{P}^j\mcD_i \mathcal{P}_j
%\;.
%\end{equation}

\subsubsection{Vanishing of the transverse derivative of  the MST on~$\Sigma$}

We would like to derive an expression, analog to \eq{Equation2}, for the anti-symmetric part of the covariant derivative of $\mathcal{P}$.
The anti-symmetric part of \eq{rep_deriv_P2}, though, was used to obtain an expression for the candidate field $Y$, whence it does not seem
to be usable for this.

Later on we will be interested in initial data $(\Sigma,h_{ij}, K_{ij})$ for which we do not know whether they admit a solution $(\sigma, Y^i)$ of the KID equations.
Instead, we want to assume that we have been given a co-vector field $\mathcal{P}$ which  solves \eq{gen_vanishing_MST}. Then the ``MST''
$ \mathcal{S}_{\alpha\beta\mu\nu}$ vanishes on $\Sigma$. The quotation marks are to emphasize that we do not know whether it is associated to  a KVF:
First of all, a solution of the KID equations does not need to exist.
And secondly, even if a solution  exists, there is a priori no reason why a solution  $\mathcal{P}$ of \eq{gen_vanishing_MST} should arise from  $(\sigma,Y^i)$ via \eq{dfn_calP}.
In both cases the ``MST'' is not the proper one,
so a priori  there is no reason to expect that  the transverse derivative of the ``MST'' vanishes on $\Sigma$ (which otherwise would follow from the fact that the MST satisfies
%, at least for an appropriately chosen $Q$,
the symmetric hyperbolic system \eq{phys_ev}).
%In fact, since the derivation of the KID equations requires to take transverse derivatives of xx\tim{add} into account,
%one might  expect that  $ \mathcal{S}_{\alpha\beta\mu\nu}|_{\Sigma}=0$ does not
%suffice to guarantee the existence of a solution to the KID equations which is related to $\mathcal{P}$ via \eq{dfn_calP}.
For this reason, it seems  promising to analyze the vanishing of the transverse derivative of the MST.
 Clearly,  relations  obtained this way \emph{necessarily} need to be fulfilled by Cauchy  data  $(\Sigma,h_{ij}, K_{ij})$
which yield a $\Lambda$-vacuum space-time with vanishing MST.
%\tim{from now on do not assume that a solution of the KID equations exist}

For the computation of  $\nabla_t\mathcal{S}_{\alpha\beta\mu\nu}|_{\Sigma}$ we need to determine the transverse derivatives  of $Q$ and $\mathcal{Q}_{\alpha\beta\mu\nu}$.
We assume
\begin{equation}
q\mathcal{P}^2 \,\ne\, 0\;, \quad q\mathcal{P}^2 - \frac{\Lambda}{2}\,\ne\,0
\;.
\end{equation}
%We  assume that the KID equations as well as \eq{gen_vanishing_MST}-\eq{q_PDE} are satisfied, but we do \emph{not} assume that that  \eq{dfn_calP} holds.
In any vacuum space-time with vanishing MST the following relations hold \cite{mars_senovilla},
\begin{eqnarray}
\nabla_{\mu}\mathcal{F}^2 &=&  \frac{4}{3} (Q \mathcal{F}^2  +2\Lambda) X^{\alpha}\mathcal{F}_{\alpha \mu}
\;,
\label{nabla_F2}
\\
\nabla_{\mu}\mathcal{F}_{\alpha\beta} &=& QX^{\kappa}\mathcal{F}_{\kappa\mu}\mathcal{F}_{\alpha\beta}
+\frac{1}{3} (Q\mathcal{F}^2 - 4\Lambda )X^{\nu}\mathcal{I}_{\alpha\beta\mu\nu}
\;.
\label{nabla_F}
\end{eqnarray}
%
%We  employ \eq{PDE_Q} to obtain an expression for the transverse derivative of $Q$,
%\tim{eqn compatible for $\mathcal{S}=0|_{\Sigma}$ with any choice of the above $Q$'s, and then we regard the equation below as definition of $\nabla_tQ|_{\Sigma}$,... more details}
%\tim{most convenient way to determine derivative...}
% as needed to determine the transverse derivative of the MST.
%
We employ  \eq{PDE_Q} and \eq{nabla_F2} to calculate
\begin{eqnarray}
\nabla_{t}Q|_{\Sigma}
% &=& -\frac{1}{4}\ol{\frac{Q\mathcal{F}^2 + 20 \Lambda}{Q\mathcal{F}^2 + 2\Lambda}}\ol{ Q\mathcal{F}^{-2}\nabla_{t}\mathcal{F}^2}
%\\
&=&
 -\frac{1}{3}(\ol Q\ol{\mathcal{F}}^2 + 20 \Lambda) \ol Q\ol{\mathcal{F}}^{-2}\ol  X^{\alpha}\ol{\mathcal{F}}_{\alpha t}
\\
&=&
 \frac{1}{3}(q\mathcal{P}^2 - 5 \Lambda) q\mathcal{P}^{-2} Y^{k}\mathcal{P}_{k }
\;.
\end{eqnarray}
%
%where we have set
%%
%\begin{equation}
%q\,:=\, Q|_{\Sigma}
%\;.
%\end{equation}
%
Finally, using \eq{trans_weyl}, \eq{nabla_F2} and \eq{nabla_F} we determine
 the transverse derivative of the MST on $\Sigma$,
\begin{eqnarray}
\nabla_t\mathcal{S}_{titj}|_{\Sigma} &\equiv &  \ol{ \nabla_t\mathcal{C}_{titj}}  -\ol{\nabla_t Q}\Big(
\ol{\mathcal{F}}_{ti}\ol{\mathcal{F}}_{tj} -\frac{\ol{\mathcal{F}}^2}{3}\ol{\mathcal{I}}_{titj}
\Big)
 - \ol{Q}\Big(2  \ol{\mathcal{F}}_{t(i}\ol{\nabla_{|t}\mathcal{F}_{t|j)}} - \frac{\ol{\nabla_t\mathcal{F}}^2}{3}\ol{\mathcal{I}}_{titj} \Big)
\nonumber
\\
 &=&  i\mathring \epsilon_{(i}{}^{kl}\mcD_{|k|}\mathcal{E}_{j)l} + K_{(i}{}^k\mathcal{E}_{j)k}  -K\mathcal{E}_{ij} + \mathring\epsilon_{i}{}^{kp}\mathring\epsilon_{j}{}^{lq}
K_{kl}\mathcal{E}_{pq}
\nonumber
\\
&&
  + \frac{1}{3}q(q\mathcal{P}^2 +  \Lambda) \Big(5\mathcal{P}^{-2} Y^{k}\mathcal{P}_{k } \mathcal{P}_{i}\mathcal{P}_{j}
-2   \mathcal{P}_{(i}Y_{j)}
- Y^{k}\mathcal{P}_{k} h_{ij}
\Big)
\;.
\phantom{xxx}
\end{eqnarray}
Assume now that the MST vanishes \emph{initially}, i.e.\ $\mathcal{S}_{\alpha\beta\mu\nu}|_{\Sigma}=0$, or, equivalently, that \eq{gen_vanishing_MST} holds
%\tim{!!!!!!!}
%%
%\begin{equation}
%\mathcal{E}_{ij} \,=\, q(\mathcal{P}_i\mathcal{P}_j)\breve{}
%\end{equation}
%%
for a function $q=q(x^i)$ which satisfies \eq{q_PDE}.
Then
\begin{eqnarray}
\nabla_t\mathcal{S}_{titj}|_{\Sigma}
 &=&  \frac{i}{4} \frac{q\mathcal{P}^2 -5 \Lambda}{q\mathcal{P}^2 - \frac{\Lambda}{2}} q\mathring \epsilon_{(i}{}^{kl} \mathcal{P}_{j)}\mathcal{P}_k  \mcD_l\log \mathcal{P}^2
+   iq\mathring \epsilon_{(i}{}^{kl}  \mathcal{P}_{|l} \mcD_{k|}\mathcal{P}_{j)}
+   i q \mathcal{P}_{(i}\mathcal{Q}_{j)}
\nonumber
\\
&&
+ q \mathring\epsilon_{i}{}^{kp}\mathring\epsilon_{j}{}^{lq}K_{kl}  \mathcal{P}_p\mathcal{P}_q
 +  q K_{(i}{}^k\mathcal{P}_{j)}\mathcal{P}_k
 - q K \mathcal{P}_i\mathcal{P}_j
\nonumber
\\
&&
  + \frac{1}{3}q(q\mathcal{P}^2 +  \Lambda) \Big(5\mathcal{P}^{-2} Y^{k}\mathcal{P}_{k } \mathcal{P}_{i}\mathcal{P}_{j}
-2   \mathcal{P}_{(i}Y_{j)}
- Y^{k}\mathcal{P}_{k} h_{ij}
\Big)
\;.
\end{eqnarray}
We plug in the expression \eq{candidate_field_Y} we derived  for $Y$,
\begin{eqnarray}
\nabla_t\mathcal{S}_{titj}|_{\Sigma}
 &=&  \frac{i}{4} \frac{4q\mathcal{P}^2 -5 \Lambda}{q\mathcal{P}^2 - \frac{\Lambda}{2}}
 q\mathring \epsilon_{(i}{}^{kl} \mathcal{P}_{j)}\mathcal{P}_k  \mcD_l\log \mathcal{P}^2
+  \frac{i}{2} q\Big( \mathring \epsilon_{i}{}^{kl}  \mathcal{P}_{l} \mcD_{(k}\mathcal{P}_{j)}
+ \mathring \epsilon_{j}{}^{kl}  \mathcal{P}_{l} \mcD_{(k}\mathcal{P}_{i)}
\Big)
\nonumber
\\
&&
+ \frac{5}{2}  i q \Big(\mathcal{Q}_{(i}   -  \mathcal{P}^{-2}\mathcal{P}^k \mathcal{Q}_k\mathcal{P}_{(i}\Big)\mathcal{P}_{j)}
+ q \mathring\epsilon_{i}{}^{kp}\mathring\epsilon_{j}{}^{lq}K_{kl}  \mathcal{P}_p\mathcal{P}_q
 + 2 q K_{(i}{}^k\mathcal{P}_{j)}\mathcal{P}_k
\nonumber
\\
&&
   - \frac{1}{2}q\Big(5  \mathcal{P}^{-2} K^{kl}\mathcal{P}_k\mathcal{P}_{l}-K\Big)\mathcal{P}_{i}\mathcal{P}_{j}
  + \frac{1}{2}q\Big(K^{kl}\mathcal{P}_k\mathcal{P}_{l}-  K \mathcal{P}^2\Big)h_{ij}
\;.
\label{trans_MST_1}
\end{eqnarray}
Contracting this with $\mathcal{P}^j$ yields
\begin{eqnarray}
\mathcal{P}^j\nabla_t\mathcal{S}_{titj}|_{\Sigma}
 &=&
 i q \mathcal{P}^2(\mathcal{Q}_{i}   - \mathcal{P}^{2}\mathcal{P}^k \mathcal{Q}_k\mathcal{P}_{i}) +  \frac{i}{4} \frac{q\mathcal{P}^2 -2 \Lambda}{q\mathcal{P}^2 - \frac{\Lambda}{2}}
 q\mathcal{P}^2\mathring \epsilon_{i}{}^{kl} \mathcal{P}_k  \mcD_l \log\mathcal{P}^2
\nonumber
\\
&&  +  q\mathcal{P}^2 K_{i}{}^k\mathcal{P}_k
   - q  K^{kl}\mathcal{P}_k\mathcal{P}_{l}\mathcal{P}_i
\;.
\end{eqnarray}
Vanishing of $\mathcal{P}^j\nabla_t\mathcal{S}_{titj}|_{\Sigma}=0$  requires
%(note that $q \mathcal{P}^2\ne 0$  in our current setting)
%\tim{!!}
% (set $\mathcal{R}_i := \mathcal{Q}_i-\mathcal{P}^{-2}\mathcal{P}^k \mathcal{Q}_k\mathcal{P}_{i} $)
%
\begin{equation}
\mathcal{Q}_{i} -  \mathcal{P}^{-2} \mathcal{P}^k \mathcal{Q}_k\mathcal{P}_{i}
\,=\,  -  \frac{1}{4} \frac{q\mathcal{P}^2 -2 \Lambda}{q\mathcal{P}^2 - \frac{\Lambda}{2}}
\mathring \epsilon_{i}{}^{kl} \mathcal{P}_k  \mcD_l \log\mathcal{P}^2
  +i   K_{i}{}^k\mathcal{P}_k
   -i \mathcal{P}^{-2}  K^{kl}\mathcal{P}_k\mathcal{P}_{l}\mathcal{P}_i
\;,
\label{Equation1B}
\end{equation}
which yields the desired relation for the anti-symmetric part of  $\nabla_i\mathcal{P}_j$.
%
% (set $\mathcal{R}_i := \mathcal{Q}_i-\mathcal{P}^{-2}\mathcal{P}^k \mathcal{Q}_k\mathcal{P}_{i} $)
%\begin{equation*}
%\boxed{
%\mathcal{R}_{i}
% \,=\,
% - \frac{1}{4} \frac{q\mathcal{P}^2 -2 \Lambda}{q\mathcal{P}^2 - \frac{\Lambda}{2}}
%\mathring \epsilon_{i}{}^{kl} \mathcal{P}_k  \mcD_l \log\mathcal{P}^2
%  + i  K_{i}{}^k\mathcal{P}_k
%   - i\mathcal{P}^{-2}  K^{kl}\mathcal{P}_k\mathcal{P}_{l}\mathcal{P}_i
%}
%\end{equation*}
%
%\begin{equation*}
%\boxed{
% \mathcal{R}_i
%\,=\,
%- \frac{4}{3}  \frac{q\mathcal{P}^2 -2 \Lambda}{q\mathcal{P}^2}\mathcal{P}^{-2}\mathcal{P}^k \mathcal{B}_{ik}
%  +\frac{4}{3} i   \frac{q\mathcal{P}^2 - \frac{\Lambda}{2}}{q\mathcal{P}^2 }\Big(  K_{i}{}^k
%   - \mathcal{P}^{-2}  K^{kl}\mathcal{P}_{l}\mathcal{P}_i \Big)\mathcal{P}_k
%}
%\end{equation*}
%
%%%%%%%%%%%%%%%%%%%%
%In terms of the anti-symmetric derivative of $\mathcal{P}$ this becomes
%\tim{needed?}
%%
%\begin{equation}
%\mcD_{[i} \mathcal{P}_{j]}\,=\,
%-\frac{1}{4}\frac{  q\mathcal{P}^{2}-2\Lambda }{  q\mathcal{P}^{2} - \frac{\Lambda}{2}}
%\mathcal{P}_{[i}\mcD_{j]}\log \mathcal{P}^2
%+\frac{1}{2}\mathring\epsilon_{ij}{}^k\Big(\mathring\epsilon^{pql}\mathcal{P}^{-2}\mathcal{P}_{k}\mathcal{P}_p\mcD_q \mathcal{P}_l
%+i K_{k}{}^{l}\mathcal{P}_{l} - i\mathcal{P}^{-2}   K^{pq}\mathcal{P}_p\mathcal{P}_{q}\mathcal{P}_{k}
%\Big)
%\;.
%\label{Equation1}
%\end{equation}
%
If one  inserts \eq{Equation1B} and \eq{Equation2}  into \eq{trans_MST_1},
%
%\begin{eqnarray}
%\nabla_t\mathcal{S}_{titj}|_{\Sigma}
% &=&  \frac{3}{8}i  \frac{q\mathcal{P}^2 }{q\mathcal{P}^2 - \frac{\Lambda}{2}}q\mathring \epsilon_{(i}{}^{kl} \mathcal{P}_{j)}\mathcal{P}_k  \mcD_l\log \mathcal{P}^2
%+   \frac{i}{2}q\Big(\mathring \epsilon_{i}{}^{kl}  \mathcal{P}_{l} \mcD_{(k}\mathcal{P}_{j)}
%+  \mathring \epsilon_{j}{}^{kl}  \mathcal{P}_{l} \mcD_{(k}\mathcal{P}_{i)} \Big)
%\nonumber
%\\
%&&
%+ q \mathring\epsilon_{i}{}^{kp}\mathring\epsilon_{j}{}^{lq}K_{kl}  \mathcal{P}_p\mathcal{P}_q
%  +\frac{q}{2} \Big(  K\mathcal{P}_i\mathcal{P}_j-     K_{(i}{}^k\mathcal{P}_{j)}\mathcal{P}_k \Big)
%\nonumber
%\\
%&&
%+\frac{q}{2}  \Big(     K^{kl}\mathcal{P}_k\mathcal{P}_{l}
% -  K\mathcal{P}^{2} \Big)h_{ij}
%\;.
%\end{eqnarray}
%
the right-hand side vanishes automatically, so no additional relation can be extracted from
$\nabla_t\mathcal{S}_{\alpha\beta\mu\nu}|_{\Sigma}=0$.

\subsubsection{An intermediate result}

It follows from \eq{gen_vanishing_MST} that
%\tim{add sth?}
\vspace{-0.30em}
\begin{equation}
q\mathcal{P}^2\,=\, \mp \sqrt{\frac{3}{2}}\sqrt{\mathcal{E}^2}
\;,
\label{relation_qP_E}
\end{equation}
whence our assumptions on $q\mathcal{P}^2$,
\vspace{-0.25em}
\begin{equation}
 q\mathcal{P}^2 \,\ne \,0
\;, \quad
q\mathcal{P}^2 + \Lambda \,\ne \, 0
\;, \quad
q\mathcal{P}^2 -\frac{\Lambda}{2} \,\ne \, 0
\;, \quad
q\mathcal{P}^2  -  2 \Lambda \,\ne \, 0
\;.
\label{all_ineqs_toget2}
\end{equation}
 can be expressed in terms of $\mathcal{E}^2$,  \eq{E_ineqs} below, i.e.\  in terms of the Cauchy data.

Let us collect the equations we have found in the preceding sections. Taking also Corollary~\ref{cor_vanishing_MST} into account we end up with the following
\begin{proposition}
\label{interm_prop_1}
Consider Cauchy data $(\Sigma, h_{ij}, K_{ij})$ which satisfy the vacuum constraint equations and%
\footnote{
\label{footnote_E_conditions}
A similar comment as in Remark~\ref{remark_weaker} applies: It is actually sufficient when the following conditions
 are satisfied for one sign $\pm$, depending on the sign
for which $\mathcal{S}^{(\mathcal{C})}_{\alpha\beta\gamma\delta}$  is fulfilled,
$$
\mathcal{E}^2 \,\ne\, 0\;, \quad
\pm\sqrt{\mathcal{E}^2} + \sqrt{\frac{1}{6}} \,\Lambda  \,\ne\, 0\;, \quad
\pm\sqrt{\mathcal{E}^2}- \sqrt{\frac{2}{3}} \,\Lambda \,\ne\, 0\;, \quad
\pm\sqrt{\mathcal{E}^2} +  \sqrt{\frac{8}{3}} \,\Lambda  \,\ne\, 0
\,.
$$
}
\begin{equation}
\mathcal{E}^2 \,\ne\, 0\;, \quad
\mathcal{E}^2- \frac{1}{6}\,\Lambda^2  \,\ne\, 0\;, \quad
\mathcal{E}^2- \frac{2}{3} \,\Lambda^2 \,\ne\, 0\;, \quad
\mathcal{E}^2 - \frac{8}{3} \,\Lambda^2  \,\ne\, 0
\;.
\label{E_ineqs}
%\label{all_ineqs_toget}
\end{equation}
A necessary condition for the emerging Cauchy development to admit a (possibly complex) KVF $X$  such that the associated MST vanishes is:
\begin{enumerate}
\item[(i)] There exists
a function $q: \Sigma \rightarrow \mathbb{C}$ and a co-vector field $\mathcal{P}$ such that
\eq{gen_vanishing_MST}, \eq{q_PDE},  \eq{Equation2}, and \eq{Equation1B} hold.
\item[(ii)] $X^{\mu}|_{\Sigma}= (\sigma, Y^i)$ is given by \eq{cand_sigma} and \eq{candidate_field_Y}.
\end{enumerate}
If, in addition to (i)-(ii),
\begin{enumerate}
\item[(iii)] $ (\sigma, Y^i)$  satisfies the KID equations, and
\item[(iv)] $(\sigma,Y^i)$ and  $\mathcal{P}_i$  are related via \eq{dfn_calP},
\end{enumerate}
then the Cauchy development of  $(\Sigma, h_{ij}, K_{ij})$ admits a
 KVF $X$ with $X^{\mu}|_{\Sigma}= (\sigma, Y^i)$  such that the associated MST vanishes in some neighborhood of $\Sigma$.
%\tim{add proof. second part different choice of $q$...}
\end{proposition}
\begin{proof}
The first part follows directly from the considerations above.
(iii)-(iv) guarantee that the tensor $\mathcal{S}_{\alpha\beta\mu\nu}$, whose vanishing on the Cauchy surface $\Sigma$ is ensured by (i), is actually the MST associated to the KVF $X$ generated by $(\sigma,Y^i)$.
The result follows now from Lemma~\ref{lemma_dq_qC} and Corollary~\ref{cor_vanishing_MST}.
\qed
\end{proof}

\subsection{The KID equations}

Let us  analyze to what extent (iii)-(iv) follow from (i)-(ii).
We consider Cauchy data $(\Sigma, h_{ij}, K_{ij})$ which satisfy the vacuum constraint equations, \eq{spec_ineqs} and  \eq{spec_ineqs2}.
Moreover, we assume that there exist a  function $q: \Sigma \rightarrow \mathbb{C}$ and a co-vector field $\mathcal{P}$ such that
\eq{gen_vanishing_MST}, \eq{q_PDE},  \eq{Equation2}, and \eq{Equation1B} hold.
Finally, we define a  (possibly complex) function $\sigma: \Sigma \rightarrow \mathbb{C}$ via  \eq{cand_sigma},  and a (possibly complex) vector field $Y$
via \eq{candidate_field_Y}.

Using \eq{q_PDE} we find that
\begin{equation}
\mcD_k\sigma \,=\,
\frac{3}{4} \Big(  q\mathcal{P}^2 - \frac{\Lambda}{2}\Big)^{-1}\Big(\mathcal{P}^l\mcD_k \mcD_l\log \mathcal{P}^2
+\mcD_k\mathcal{P}^l \mcD_l\log \mathcal{P}^2
 - \sigma\frac{q\mathcal{P}^2 + \Lambda}{q\mathcal{P}^2 - \frac{\Lambda}{2}} q\mathcal{P}^2\mcD_k\log \mathcal{P}^2
\Big)
\;.
\label{deriv_sigma}
\end{equation}
Differentiating \eq{Equation2} yields with \eq{q_PDE}
\begin{eqnarray}
\mcD_k \mcD_{(i} \mathcal{P}_{j)}
%&=&
%  -\frac{1}{3}\mcD_k \sigma(q\mathcal{P}^2 + \Lambda) h_{ij}
%  - \frac{1}{4}\sigma \frac{q\mathcal{P}^2 + \Lambda}{q\mathcal{P}^2 - \frac{\Lambda}{2}} q\mathcal{P}^2 \mcD_k\log \mathcal{P}^2  h_{ij}
%\\
%&&
%-\frac{9}{32} \Lambda   \frac{q\mathcal{P}^2 + \Lambda}{ (q \mathcal{P}^2-\frac{\Lambda}{2})^3} q\mathcal{P}^2\mathcal{P}_{(i}\mcD_{j)} \log\mathcal{P}^2
% \mcD_k\log \mathcal{P}^2
%\\
%&&
%+\frac{3}{4}  \frac{q\mathcal{P}^2}{  q \mathcal{P}^2-\frac{\Lambda}{2}} \mcD_k\mathcal{P}_{(i}\mcD_{j)} \log\mathcal{P}^2
%+\frac{3}{4}  \frac{q\mathcal{P}^2}{  q \mathcal{P}^2-\frac{\Lambda}{2}} \mathcal{P}_{(i}\mcD_{j)} \mcD_k\log\mathcal{P}^2
%\\
%&&
%+ i \mathring\epsilon_{(i}{}^{pq} \mathcal{P}_q\mcD_k K_{j)p}
%+ i \mathring\epsilon_{(i}{}^{pq} K_{j)p}\mcD_k\mathcal{P}_q
%\\
&=&
 \frac{1}{4} \frac{q\mathcal{P}^2 + \Lambda}{q\mathcal{P}^2-\frac{\Lambda}{2}}\Big(
 \frac{3}{2}\frac{\Lambda}{q\mathcal{P}^2-\frac{\Lambda}{2}} \sigma q\mathcal{P}^2\mcD_k\log \mathcal{P}^2
- \mcD_k\mathcal{P}^l \mcD_l\log \mathcal{P}^2
 - \mathcal{P}^l\mcD_k \mcD_l\log \mathcal{P}^2
\Big) h_{ij}
\nonumber
\\
&&
-\frac{9}{32} \Lambda   \frac{q\mathcal{P}^2 + \Lambda}{ (q \mathcal{P}^2-\frac{\Lambda}{2})^3} q\mathcal{P}^2\mathcal{P}_{(i}\mcD_{j)} \log\mathcal{P}^2
 \mcD_k\log \mathcal{P}^2
+\frac{3}{4}  \frac{q\mathcal{P}^2}{  q \mathcal{P}^2-\frac{\Lambda}{2}} \mathcal{P}_{(i}\mcD_{j)} \mcD_k\log\mathcal{P}^2
\nonumber
\\
&&
+\frac{3}{4}  \frac{q\mathcal{P}^2}{  q \mathcal{P}^2-\frac{\Lambda}{2}} \mcD_k\mathcal{P}_{(i}\mcD_{j)} \log\mathcal{P}^2
+ i \mathring\epsilon_{(i}{}^{pq}\Big( \mathcal{P}_q\mcD_k K_{j)p}
+K_{j)p}\mcD_k\mathcal{P}_q
\Big)
\;.
\label{2nd_order}
\end{eqnarray}
For the covariant derivative of $Y$ a somewhat lengthy calculation, which makes use of   \eq{2nd_order}, \eq{Equation2}, and the vacuum constraints,
reveals that
\begin{eqnarray}
\mcD_iY_j
&=&
\frac{3}{2} \Big( q\mathcal{P}^{2} + \Lambda\Big)^{-1}\Big[
\frac{3}{4} \frac{q\mathcal{P}^2}{q\mathcal{P}^2 - \frac{\Lambda}{2}}  \mcD_i \log\mathcal{P}^2 \Big( K_{j}{}^{k}\mathcal{P}_{k}
- K \mathcal{P}_{j}
\Big)
\nonumber
\\
&&
+i\mathring\epsilon_j{}^{kl}
\Big(
h_{ik}(2\mathring R_{lm}\mathcal{P}^m-\mathring R  \mathcal{P}_l) - 2\mathring R_{il} \mathcal{P}_k
+
\frac{3}{4} \frac{ q\mathcal{P}^2 }{q\mathcal{P}^2 - \frac{\Lambda}{2}} \mcD_k \mathcal{P}_l\mcD_i \log\mathcal{P}^2
\nonumber
\\
&&+\frac{3}{4} \frac{ 1}{  q\mathcal{P}^{2} - \frac{\Lambda}{2}} \Big(1+\frac{3}{4}   \frac{ q\mathcal{P}^{2}}{ q\mathcal{P}^{2} - \frac{\Lambda}{2}}
+ \frac{3}{8}\Lambda\frac{q\mathcal{P}^2 + \Lambda}{(q\mathcal{P}^2 - \frac{\Lambda}{2})^2}
\Big) q\mathcal{P}^{2}\mathcal{P}_k\mcD_l\log\mathcal{P}^2\mcD_i \log\mathcal{P}^2
\nonumber
\\
&&
- \frac{3}{4} \frac{ q\mathcal{P}^{2}}{  q\mathcal{P}^{2} - \frac{\Lambda}{2}}\mcD_i\mathcal{P}_k\mcD_l\log\mathcal{P}^2
 - \frac{3}{4} \frac{ q\mathcal{P}^{2}}{  q\mathcal{P}^{2} - \frac{\Lambda}{2}}\mathcal{P}^{-2}\mathcal{P}_k\mcD_i\mcD_l\mathcal{P}^2
- 2\mcD_k\mcD_{(i}  \mathcal{P}_{l)}
\Big)
\nonumber
\\
&&
 -\mathcal{P}_{k} \mcD_i K_{j}{}^{k}+\mathcal{P}_{j}\mcD_i K
-K_{j}{}^{k}\mcD_i\mathcal{P}_{k}+  K \mcD_i\mathcal{P}_{j}
\Big]
\\
&=&
\frac{3}{2} \Big( q\mathcal{P}^{2} + \Lambda\Big)^{-1}\Big[
\frac{3}{4} \frac{q\mathcal{P}^2}{q\mathcal{P}^2 - \frac{\Lambda}{2}} ( K_{j}{}^{k}\mathcal{P}_{k}
- K \mathcal{P}_{j})  \mcD_i \log\mathcal{P}^2
\nonumber
\\
&&+\frac{9}{16}  i \frac{ (q\mathcal{P}^{2})^2}{ (q\mathcal{P}^{2} - \frac{\Lambda}{2})^2}
\mathring\epsilon_j{}^{kl}\mathcal{P}_k\mcD_l\log\mathcal{P}^2\mcD_i \log\mathcal{P}^2
- \frac{3}{2}i \frac{ q\mathcal{P}^{2}}{  q\mathcal{P}^{2} - \frac{\Lambda}{2}}\mathring\epsilon_j{}^{kl}\mcD_{(i}\mathcal{P}_{k)}\mcD_l\log\mathcal{P}^2
\nonumber
\\
&&
+\mathring\epsilon_j{}^{kl}\mathring\epsilon_{i}{}^{pq} K_{lp}\mcD_k\mathcal{P}_q
 -2\mathring\epsilon_{i}{}^{kl}B_{jk}\mathcal{P}_l
- \Lambda K_{ij}\sigma
- K_{i}{}^k\mcD_k\mathcal{P}_j -K_{j}{}^{k}\mcD_i\mathcal{P}_{k}+  K \mcD_i\mathcal{P}_{j}
\nonumber
\\
&&
%-\mathring\epsilon_{[i}{}^{kl} B_{j]k}\mathcal{P}_l
%+\frac{1}{2} \mathring \epsilon_{ij}{}^k B_{kl}\mathcal{P}^{l}
 + 2i\mathring\epsilon_j{}^{kl}\mathring R_{ik} \mathcal{P}_l
-i \mathring\epsilon_{ij}{}^{k}(2\mathring R_{kl}\mathcal{P}^l-\mathring R  \mathcal{P}_k)
\Big]
\nonumber
\\
&&
-  \frac{3}{4}i \frac{1}{q\mathcal{P}^2-\frac{\Lambda}{2}}\mathring\epsilon_{ij}{}^{k}\Big(
 \frac{3}{2}\frac{\Lambda}{q\mathcal{P}^2-\frac{\Lambda}{2}} \sigma q\mathcal{P}^2\delta_k{}^l
- \mcD_k\mathcal{P}^l
 - \mathcal{P}^l\mcD_k
\Big) \mcD_l\log \mathcal{P}^2
\\
&=&
\frac{3}{2} \Big( q\mathcal{P}^{2} + \Lambda\Big)^{-1}\Big[
\frac{3}{4} \frac{q\mathcal{P}^2}{q\mathcal{P}^2 - \frac{\Lambda}{2}} \Big(
  K_{i}{}^k\mathcal{P}_j\mcD_k\log\mathcal{P}^2
+( K_{j}{}^{k}\mathcal{P}_{k}- K \mathcal{P}_{j})\mcD_i \log\mathcal{P}^2
\Big)
\nonumber
\\
&&
+\mathring\epsilon_j{}^{kl}\mathring\epsilon_{i}{}^{pq} K_{lp}\Big(
\mcD_k\mathcal{P}_q
- \frac{3}{4} \frac{ q\mathcal{P}^{2}}{  q\mathcal{P}^{2} - \frac{\Lambda}{2}}\mathcal{P}_q\mcD_k\log\mathcal{P}^2
\Big)
 -2\mathring\epsilon_{i}{}^{kl}B_{jk}\mathcal{P}_l
- K_{i}{}^k\mcD_k\mathcal{P}_j
\nonumber
\\
&&
 -K_{j}{}^{k}\mcD_i\mathcal{P}_{k}
+  K \mcD_i\mathcal{P}_{j}
 + 2i\mathring\epsilon_j{}^{kl}\mathring R_{ik} \mathcal{P}_l
-i \mathring\epsilon_{ij}{}^{k}(2\mathring R_{kl}\mathcal{P}^l-\mathring R  \mathcal{P}_k)
\Big] - \frac{3}{2}K_{ij}\sigma
\nonumber
\\
&&
\underbrace{-  \frac{3}{4}i \frac{1}{q\mathcal{P}^2-\frac{\Lambda}{2}}\mathring\epsilon_{ij}{}^{k}\Big(
 \frac{q\mathcal{P}^2 + \Lambda}{q\mathcal{P}^2-\frac{\Lambda}{2}} \sigma q\mathcal{P}^2\delta_k{}^l
- \mcD_k\mathcal{P}^l
 - \mathcal{P}^l\mcD_k
\Big)\mcD_l\log \mathcal{P}^2 }_{=i\mathring\epsilon_{ij}{}^k\mcD_k\sigma}
\;.
\label{full_deriv_Y}
\end{eqnarray}

Now we are ready to consider the KID equation \eq{KID1}.
Taking the symmetric, part of  \eq{full_deriv_Y} another lengthy calculation shows  (we use  \eq{2nd_order}, \eq{Equation2}, \eq{Equation1B},  \eq{gen_vanishing_MST} and the vacuum constraints),
%(With above 2nd-order derivative relation) that yields,
%divP-relation, constraints,
%
\begin{eqnarray}
\mcD_{(i}Y_{j)}
 &=&
\frac{3}{2} \Big( q\mathcal{P}^{2} + \Lambda\Big)^{-1}\Big[
\frac{3}{4} \frac{q\mathcal{P}^2}{q\mathcal{P}^2 - \frac{\Lambda}{2}} \Big( K_{(i}{}^{k}\mathcal{P}_{k}
- K \mathcal{P}_{(i} \Big) \mcD_{j)} \log\mathcal{P}^2
\nonumber
\\
&&
- 2i\mathring\epsilon_{(i}{}^{kl} \mathring R_{j)l} \mathcal{P}_k
- \frac{3}{2}i \frac{ q\mathcal{P}^2 }{q\mathcal{P}^2 - \frac{\Lambda}{2}}
\mathring\epsilon_{(i}{}^{kl} \mcD_{(j)}\mathcal{P}_{k)}\mcD_l\log\mathcal{P}^2
\nonumber
\\
&&+\frac{9}{16} i\frac{ (q\mathcal{P}^{2})^2}{  (q\mathcal{P}^{2} - \frac{\Lambda}{2})^2}
\mathring\epsilon_{(i}{}^{kl}\mathcal{P}_k\mcD_{j)} \log\mathcal{P}^2 \mcD_l\log\mathcal{P}^2
-\Lambda\sigma K_{ij}
-2\mathring\epsilon_{(i}{}^{pq}B_{j)p} \mathcal{P}_q
\nonumber
\\
&&
+ \mathring\epsilon_{(i}{}^{kl}\mathring\epsilon_{j)}{}^{pq} K_{lp}\mcD_{(k}\mathcal{P}_{q)}
-2K_{(i}{}^{k}\mcD_{(j)}\mathcal{P}_{k)}+  K \mcD_{(i}\mathcal{P}_{j)}
\Big]
\\
 &=&
\frac{3}{2} i\Big( q\mathcal{P}^{2} + \Lambda\Big)^{-1}\Big[
\underbrace{  \mathring\epsilon_{(i}{}^{kl}\Big(K_{j)k}K_{lp} \mathcal{P}^p
+K_{j)}{}^{p} K_{pk}\mathcal{P}_l
-  K K_{j)k}\mathcal{P}_l
\Big)}_{=0}
\nonumber
\\
&& +  2\mathring\epsilon_{(i}{}^{kl}\mathcal{ E}_{j)k}   \mathcal{P}_l
\Big]
- \sigma K_{ij}
\\
&=&
- \sigma K_{ij}
\;.
\end{eqnarray}
%
%For the last step  we  employed \eq{gen_vanishing_MST}.
The first KID equation \eq{KID1} is therefore automatically fulfilled in this setting.

Before we analyze the second KID equation, it is useful to focus attention to another equation first, namely
\eq{dfn_calP},
\begin{equation}
\mathcal{P}_{i}
\,=\,  \mcD_{i}\sigma +K_{ij}Y^j + \frac{i}{2} Z_i
\;.
\label{dfn_calP2}
\end{equation}
%
%We need to make sure that
%%
%\begin{eqnarray*}
%\mathcal{F}_{\alpha\beta}
%&=& 2\mathcal{I}_{\alpha\beta}{}^{\mu\nu}\nabla_{\mu}X_{\nu}
%\end{eqnarray*}
%%
%Evaluated on $\Sigma$ this becomes
%
It ensures that $\mathcal{P}$ and $(\sigma,Y)$ are related in the right way, so that the  ``MST'',~given on $\Sigma$ in terms of $\mathcal{P}$, is actually the proper MST associated to a KVF, namely the one generated by $(\sigma, Yî)$.
To check it, we determine the anti-symmetric part of  \eq{full_deriv_Y}.
With \eq{Equation1B},  \eq{deriv_sigma}, \eq{gen_vanishing_MST} and the vacuum constraints we find
%
%antisym deriv, deriv sigma, dfnE, dfnB, constraint
%PP-E

\begin{eqnarray*}
\mcD_{[i}Y_{j]}
&=&
\frac{3}{2} \Big( q\mathcal{P}^{2} + \Lambda\Big)^{-1}\Big[
\frac{3}{4} \frac{q\mathcal{P}^2}{q\mathcal{P}^2 - \frac{\Lambda}{2}} \Big(
  K_{[i}{}^k\mathcal{P}_{j]}\mcD_{k}\log\mathcal{P}^2
- (  K_{[i}{}^k\mathcal{P}_{|k|} -    K \mathcal{P}_{[i}) \mcD_{j]}\log\mathcal{P}^2
\Big)
\\
&&
-\mathring\epsilon_{[i}{}^{kl}\mathring\epsilon_{j]}{}^{pq} K_{lp}\Big(
\mcD_{[k}\mathcal{P}_{q]}
- \frac{3}{4} \frac{ q\mathcal{P}^{2}}{  q\mathcal{P}^{2} - \frac{\Lambda}{2}}\mathcal{P}_{[q}\mcD_{k]}\log\mathcal{P}^2
\Big)
+ K_{[i}{}^{k}\mcD_{j]}\mathcal{P}_{k}
- K_{[i}{}^{k}\mcD_{|k|}\mathcal{P}_{j]}
\\
&&
+  K \mcD_{[i}\mathcal{P}_{j]}
 -2\mathring\epsilon_{[i}{}^{kl}B_{j]k}\mathcal{P}_l
 - 2i\mathring\epsilon_{[i}{}^{kl}\mathring R_{j]k} \mathcal{P}_l
-i \mathring\epsilon_{ij}{}^{k}(2\mathring R_{kl}\mathcal{P}^l-\mathring R  \mathcal{P}_k)
\Big]
+ i\mathring\epsilon_{ij}{}^k\mcD_k\sigma
\\
&=&
\frac{3}{4} \frac{1 }{q\mathcal{P}^2 - \frac{\Lambda}{2}} \Big(
  K_{[i}{}^k\mathcal{P}_{j]}\mcD_{k}\log\mathcal{P}^2
- \mathcal{P}_{k} K_{[i}{}^k\mcD_{j]}\log\mathcal{P}^2
+ K \mathcal{P}_{[i}\mcD_{j]} \log\mathcal{P}^2
\\
&&
\qquad
-\mathring\epsilon_{[i}{}^{kl}\mathring\epsilon_{j]}{}^{pq} K_{lp}
\mathcal{P}_{[k}\mcD_{q]}\log \mathcal{P}^2
\Big)
\\
&&
+\frac{3}{2} \Big( q\mathcal{P}^{2} + \Lambda\Big)^{-1}\underbrace{\frac{1}{2}\Big(\mathring\epsilon_{ij}{}^{p} K_{lp}
+2 K_{[i}{}^{k}\mathring\epsilon_{j]k}{}^l
+ K\mathring\epsilon_{ij}{}^l
\Big)}_{=\mathring\epsilon_{ij}{}^{p} K_{p}{}^l}
\\
&&
\qquad
\times
\Big(\mathcal{P}^{-2}\mathcal{P}_{l}\mathcal{P}^m \mathcal{Q}_m
+i K_{l}{}^{r}\mathcal{P}_{r} - i\mathcal{P}^{-2}   K^{mn}\mathcal{P}_m\mathcal{P}_{n}\mathcal{P}_{l}
\Big)
\\
&&
+ \frac{3}{2} \Big( q\mathcal{P}^{2} + \Lambda\Big)^{-1}\Big[
 -2i\mathring\epsilon_{[i}{}^{kl}(K K_{j]k} - K_{j]l}K_k{}^l -\frac{2}{3}\Lambda h_{j]k} )\mathcal{P}_l
 +2i\mathring\epsilon_{[i}{}^{kl}\mathcal{E}_{j]k}\mathcal{P}_l
\\
&&
\qquad
-i\underbrace{ \Big(\mathring\epsilon_{ij}{}^{k}(2\mathring R_{kl}\mathcal{P}^l-\mathring R  \mathcal{P}_k)
 + 4\mathring\epsilon_{[i}{}^{kl}\mathring R_{j]k} \mathcal{P}_l
\Big)}_{=\mathring\epsilon_{ij}{}^k\mathring R\mathcal{P}_k}
\Big]
+ i\mathring\epsilon_{ij}{}^k\mcD_k\sigma
%%%%%%%%%%%%
\end{eqnarray*}
\begin{eqnarray*}
\phantom{\mcD_{[i}Y_{j]} }&=&
\frac{3}{4} \frac{1 }{q\mathcal{P}^2 - \frac{\Lambda}{2}} \Big[
  K_{[i}{}^k\mathcal{P}_{j]}\mcD_{k}\log\mathcal{P}^2
- \mathcal{P}_{k} K_{[i}{}^k\mcD_{j]}\log\mathcal{P}^2
+ K \mathcal{P}_{[i}\mcD_{j]} \log\mathcal{P}^2
\\
&&
\qquad
-\mathring\epsilon_{[i}{}^{kl}\mathring\epsilon_{j]}{}^{pq} K_{lp}
\mathcal{P}_{[k}\mcD_{q]}\log \mathcal{P}^2
-\mathring\epsilon_{ij}{}^{p} K_{p}{}^l\mathring\epsilon_l{}^{km}\mathcal{P}_k\mcD_m\log\mathcal{P}^2
\Big]
- \frac{3}{2} i\Big( q\mathcal{P}^{2} + \Lambda\Big)^{-1}
\\
&&
\times
\underbrace{\Big[
 2\mathring\epsilon_{[i}{}^{kl}(K K_{j]k} - K_{j]m}K_k{}^m  )\mathcal{P}_l
+ \mathring\epsilon_{ij}{}^p[ (|K|^2-K^2)\mathcal{P}_p
-  K_{p}{}^l(K_{l}{}^{k}\mathcal{P}_{k} -K \mathcal{P}_{l})]
\Big] }_{=0}
\\
&&
+ i\mathring\epsilon_{ij}{}^{k} ( \mcD_k\sigma -\mathcal{P}_k + K_{kl}Y^l)
\;.
\end{eqnarray*}
Recall our definition \eq{candidate_field_Y} of $Y$. With \eq{Equation1B}
it can be written as
\begin{eqnarray*}
Y_i
%&=&
%-\frac{3}{2} \Big( q\mathcal{P}^{2} + \Lambda\Big)^{-1}\Big[\frac{3}{4}i\mathring\epsilon_i{}^{kl} \frac{ q\mathcal{P}^{2}}{  q\mathcal{P}^{2} - \frac{\Lambda}{2}}\mathcal{P}_k\mcD_l\log\mathcal{P}^2
%+ i\mathcal{Q}_i
% +K_{i}{}^{k}\mathcal{P}_{k}-  K \mathcal{P}_{i}
%\Big]
%\\
&=&
-\frac{3}{2}i \Big( q\mathcal{P}^{2} + \Lambda\Big)^{-1}\Big(\mathcal{P}^{-2}\mathcal{P}_{i}\mathcal{P}^k \mathcal{Q}_k
+i K_{i}{}^{k}\mathcal{P}_{k} - i\mathcal{P}^{-2}   K^{kl}\mathcal{P}_k\mathcal{P}_{l}\mathcal{P}_{i}
\Big)
\\
&&
-\frac{3}{4}i \Big( q\mathcal{P}^{2} - \frac{\Lambda}{2}\Big)^{-1}\mathring\epsilon_i{}^{kl}\mathcal{P}_k\mcD_l\log\mathcal{P}^2
-\frac{3}{2} \Big( q\mathcal{P}^{2} + \Lambda\Big)^{-1}\Big(K_{i}{}^{k}\mathcal{P}_{k}-  K \mathcal{P}_{i}
\Big)
\;.
\end{eqnarray*}
We insert this into the expression we have derived for $\mcD_{[i}Y_{j]}$ to end up with
\begin{equation}
\mcD_{[i}Y_{j]}
\,=\,
i\mathring\epsilon_{ij}{}^{k} ( \mcD_k\sigma -\mathcal{P}_k + K_{kl}Y^l)
\;,
\label{relation_P_Y_sigma}
\end{equation}
which is equivalent to \eq{dfn_calP2}, i.e.\ $\mathcal{P}$ and $(\sigma,Y)$ are automatically related to each other in the desired
way.
Morover, it follows  immediately  from \eq{relation_P_Y_sigma} that for $\mathcal{P}^2 \ne 0$ the emerging KVF cannot be trivial on $\Sigma$.

Finally, let us devote attention to the second KID equation \eq{KID2}.
We differentiate \eq{dfn_calP2}.
Using \eq{Equation2}, \eq{KID1}, \eq{gen_vanishing_MST}, \eq{candidate_field_Y}, \eq{Equation1B}
yields
%symderivP, 1st Killing, EP-relation,dfnY, antisymP contrc with P, dfn sigma
%
\begin{eqnarray}
\mcD_i \mcD_{j}\sigma  &=&
\mcD_{(i}\mathcal{P}_{j)} -  \frac{i}{2}\mcD_{(i}Z_{j)}- Y^k \mcD_{(i}K_{j)k}- K_{k(i}\mcD_{j)}Y^k
\\
&=&
\mcD_{(i}\mathcal{P}_{j)}
-  \frac{i}{2}(\mathring\epsilon_{i}{}^{kl} \mcD_{k}\mcD_{(j}Y_{l)}  + \mathring\epsilon_{j}{}^{kl} \mcD_{k}\mcD_{(i}Y_{l)})
+  \frac{i}{2}\mathring\epsilon_{(i}{}^{kl} \mathring R_{j)pkl}Y^p
\nonumber
\\
&&
 - Y^k \mcD_{(i}K_{j)k}- K_{(i}{}^k\mcD_{j)}Y^k
\\
&=&
  -\frac{1}{3}\sigma(q\mathcal{P}^2 + \Lambda) h_{ij}
+\frac{3}{4}  \frac{q\mathcal{P}^2}{  q \mathcal{P}^2-\frac{\Lambda}{2}} \mathcal{P}_{(i}\mcD_{j)} \log\mathcal{P}^2
 -  i\sigma B_{ij}
- Y^k \mcD_{(i}K_{j)k}
\nonumber
\\
&&
- K_{k(i}\mcD_{j)}Y^k
+ i \mathring\epsilon_{(i}{}^{kl}\Big( K_{j)k}\mathcal{P}_l
 +K_{j)l} \mcD_{k}\sigma + \mathring R_{j)k} Y_l\Big)
\\
&=&
- \sigma\Big(   i B_{ij}  +\frac{1}{3}q\mathcal{P}^2   + \frac{\Lambda}{3}  h_{ij}\Big)
+\frac{3}{4}  \frac{q\mathcal{P}^2}{  q \mathcal{P}^2-\frac{\Lambda}{2}} \mathcal{P}_{(i}\mcD_{j)} \log\mathcal{P}^2
-  \mcL_YK_{ij}
\nonumber
\\
&&
+  K_{(i}{}^k \mcD_{(j)}Y_{k)}
+ i\mathring\epsilon_{(i}{}^{kl}K_{j)k} K_{lp}Y^p
+ i\mathring\epsilon_{(i}{}^{kl} (\mathring R_{j)k} +i B_{j)k}) Y_l
\\
&=&
 \sigma\Big(
 \mathring R_{ij}  +  K K_{ij} - 2  K_{ik}K_j{}^k - \Lambda h_{ij} - \mathcal{E}_{ij} - \frac{1}{3}q\mathcal{P}^2   \Big)
\nonumber
\\
&&
+\frac{3}{4}  \frac{q\mathcal{P}^2}{  q \mathcal{P}^2-\frac{\Lambda}{2}} \mathcal{P}_{(i}\mcD_{j)} \log\mathcal{P}^2
+ i\mathring\epsilon_{(i}{}^{kl} \mathcal{E} _{j)k}  Y_l
-  \mcL_YK_{ij}
\nonumber
\\
&&
+\underbrace{ i\mathring\epsilon_{(i}{}^{kl} \Big( K_{j)k} K_{l}{}^{p}Y_p  + K_{j)p}K_k{}^pY_l -   K K_{j)k} Y_l  \Big)}_{=0}
\\
&=&
 \sigma\Big(
 \mathring R_{ij}  +  K K_{ij} - 2  K_{ik}K_j{}^k - \Lambda h_{ij}    \Big)-  \mcL_YK_{ij}
\;,
\end{eqnarray}
i.e.\ the second KID equation holds automatically, as well.

%\begin{equation*}
%\mcD_{[i} \mathcal{P}_{j]}=
%-\frac{1}{4}\frac{  q\mathcal{P}^{2}-2\Lambda }{  q\mathcal{P}^{2} - \frac{\Lambda}{2}}
%\mathcal{P}_{[i}\mcD_{j]}\log \mathcal{P}^2
%+\frac{1}{2}\mathring\epsilon_{ij}{}^s\Big(\mathcal{P}^{-2}\mathcal{P}_{s}\mathcal{P}^m \mathcal{Q}_m
%+i K_{s}{}^{r}\mathcal{P}_{r} - i\mathcal{P}^{-2}   K^{mn}\mathcal{P}_m\mathcal{P}_{n}\mathcal{P}_{s}
%\Big)
%\end{equation*}

We are thus led to the following improvement of Proposition~\ref{interm_prop_1}
\begin{proposition}
\label{interm_prop_2}
Consider Cauchy data $(\Sigma, h_{ij}, K_{ij})$ which satisfy the vacuum constraint equations and \eq{E_ineqs} (cf.\ footnote~\ref{footnote_E_conditions}).
%
%\begin{equation}
%\mathcal{E}^2 \,\ne\, 0\;, \quad
%\sqrt{\mathcal{E}^2} + \sqrt{\frac{1}{6}} \,\Lambda  \,\ne\, 0\;, \quad
%\sqrt{\mathcal{E}^2}- \sqrt{\frac{2}{3}} \,\Lambda \,\ne\, 0\;, \quad
%\sqrt{\mathcal{E}^2} + \sqrt{\frac{8}{3}} \,\Lambda  \,\ne\, 0
%\;.
%\label{E_ineqs}
%\end{equation}
The emerging Cauchy development
 admits a (possibly complex) KVF $X$  such that the associated MST vanishes if and only if
 there exists
a function $q: \Sigma \rightarrow \mathbb{C}$ and a co-vector field $\mathcal{P}$ such that
\eq{gen_vanishing_MST}, \eq{q_PDE},  \eq{Equation2} and \eq{Equation1B} hold.
In that case $X^{\mu}|_{\Sigma}= (\sigma, Y^i)$, where $\sigma$ and $Y^i$ are given by \eq{cand_sigma} and \eq{candidate_field_Y}, respectively,
and $X^{\mu}|_{\Sigma}$ is non-trivial.
\end{proposition}

\subsection{The equations for $\mathcal{P}$ revisited}

Proposition~\ref{interm_prop_2} requires the existence of  a function $q$ and a co-vector field $\mathcal{P}$ such that, for given Cauchy data $(\Sigma, h_{ij}, K_{ij}$)
(recall the definition of $\mathcal{E}_{ij}$ \eq{dfn_cal_E}, $\sigma$ \eq{cand_sigma} and  $\mathcal{Q}_i$ \eq{dfn_calQ}),
\begin{eqnarray}
\mathcal{E}_{ij} &=&  q (\mathcal{P}_i\mathcal{P}_j)\breve{}
\;,
\label{gen_vanishing_MST2}
\\
 \mcD_i q &=& - \frac{1}{4} \frac{q\mathcal{P}^2 -5 \Lambda}{q\mathcal{P}^2 - \frac{\Lambda}{2}} q \mcD_i \log\mathcal{P}^2
\;,
\label{q_PDE2}
\\
\mcD_{(i} \mathcal{P}_{j)}
&=&
\frac{3}{4}  \frac{q\mathcal{P}^2}{  q \mathcal{P}^2-\frac{\Lambda}{2}} \mathcal{P}_{(i}\mcD_{j)} \log\mathcal{P}^2
+ i \mathring\epsilon_{(i}{}^{kl} K_{j)k}\mathcal{P}_l
 -\frac{\sigma}{3}(q\mathcal{P}^2 + \Lambda) h_{ij}
\;,
\phantom{xx}
\label{Equation2_2}
\\
\mathcal{Q}_{i}
&=&  \mathcal{P}^{-2} \mathcal{P}^k \mathcal{Q}_k\mathcal{P}_{i}   -  \frac{1}{4} \frac{q\mathcal{P}^2 -2 \Lambda}{q\mathcal{P}^2 - \frac{\Lambda}{2}}
\mathring \epsilon_{i}{}^{kl} \mathcal{P}_k  \mcD_l \log\mathcal{P}^2
  +i   K_{i}{}^k\mathcal{P}_k
\nonumber
\\
&&\qquad
   -i \mathcal{P}^{-2}  K^{kl}\mathcal{P}_k\mathcal{P}_{l}\mathcal{P}_i
\;.
\label{Equation1B_2}
\end{eqnarray}
In this section  we want to analyze to what extent these equations are independent of each other, or rather if one of them is implied by the remaining ones.
The most promising  starting point is undoubtedly to differentiate  \eq{gen_vanishing_MST2}.
%We take the derivative of  \eq{gen_vanishing_MST2}.
With \eq{q_PDE2} we obtain
\begin{equation}
\mcD_k\mathcal{E}_{ij} +  \frac{1}{4} \frac{q\mathcal{P}^2 -5 \Lambda}{q\mathcal{P}^2 - \frac{\Lambda}{2}} q  \mathcal{P}_i\mathcal{P}_j\mcD_k \log\mathcal{P}^2
 - 2q \mathcal{P}_{(i} \mcD_{|k|}\mathcal{P}_{j)}
 +  \frac{1}{4} \frac{q\mathcal{P}^2 + \Lambda}{q\mathcal{P}^2 - \frac{\Lambda}{2}} q h_{ij}\mcD_k \mathcal{P}^2  \,=\,0
\;.
\label{deriv_main_eqn}
\end{equation}
To extract  one of the above equations, though,  we need to eliminate the derivative of $\mathcal{E}_{ij}$.
For this purpose, recall that the vacuum constraints impose  restrictions \eq{constraints_E} % which are fulfilled by
on $\mathcal{E}_{ij}$. Taking also \eq{gen_vanishing_MST2}
into account that yields an equation of a form we are looking for,
\begin{equation}
%h^{ij}\mathcal{E}_{ij}\,=\, 0
%\;, \quad
\mcD_j \mathcal{E}_{i}{}^{j}
 %\,=\,     i \mathring\epsilon_{i}{}^{jk}K_j{}^l \mathcal{E}_{kl}
\,=\,    i q\mathring\epsilon_{i}{}^{jk}K_j{}^l \mathcal{P}_{k}\mathcal{P}_{l}
\;.
\label{constraints_EP}
\end{equation}
On the other hand, applying $h^{jk}$ to  \eq{deriv_main_eqn}  yields
%, qP-relation, divP
\begin{equation}
\mcD_j\mathcal{E}_{i}{}^{j}
 - 2q \mathcal{P}^{j}\mcD_{(i}\mathcal{P}_{j)}
 - q \mathcal{P}_{i} \mcD_{j}\mathcal{P}^{j}
+  \frac{1}{4} \frac{q\mathcal{P}^2 -5 \Lambda}{q\mathcal{P}^2 - \frac{\Lambda}{2}} q  \mathcal{P}_i\mathcal{P}^j\mcD_j \log\mathcal{P}^2
 +  \frac{3}{4} \frac{q\mathcal{P}^2 }{q\mathcal{P}^2 - \frac{\Lambda}{2}} q \mcD_i \mathcal{P}^2  \,=\,0
\;.
\end{equation}
Combined we obtain (recall \eq{cand_sigma} and note that $q\ne 0$ in our  current setting of Proposition~\ref{interm_prop_2})
\begin{equation}
 2 \mathcal{P}^{j}\mcD_{(i}\mathcal{P}_{j)}
 +  \mathcal{P}_{i} \mcD_{j}\mathcal{P}^{j}
-  \frac{1}{3}\sigma (q\mathcal{P}^2 -5 \Lambda)  \mathcal{P}_i
 -  \frac{3}{4} \frac{q\mathcal{P}^2 }{q\mathcal{P}^2 - \frac{\Lambda}{2}}  \mcD_i \mathcal{P}^2
-  i \mathring\epsilon_{i}{}^{jk}K_j{}^l \mathcal{P}_{k}\mathcal{P}_{l}
\,=\,0
\;.
\label{interm_div_relation}
\end{equation}
We contract this equation with $\mathcal{P}^i$ to recover \eq{divP} as the trace of   \eq{Equation2_2},
\begin{equation}
 \mcD_{j}\mathcal{P}^{j}
\,=\, -\Lambda \sigma
\,.
\label{sigma_Lambda}
\end{equation}
We insert this into \eq{interm_div_relation}
to  recover \eq{Equation2_contr} as the contraction of  \eq{Equation2_2} with $\mathcal{P}^j$,
\begin{equation}
\mathcal{P}^j\mcD_{(i} \mathcal{P}_{j)}
\,=\,   \frac{1}{6}\sigma(q\mathcal{P}^2 -2 \Lambda)\mathcal{P}_i
+\frac{3}{8}  \frac{q\mathcal{P}^2}{  q \mathcal{P}^2-\frac{\Lambda}{2}}\mcD_{i} \mathcal{P}^2
+ \frac{i}{2} \mathring\epsilon_{i}{}^{jk} K_{j}{}^l\mathcal{P}_k\mathcal{P}_l
\;.
\label{Equation2_contr_2}
\end{equation}
%
%
%Finally we contract \eq{deriv_main_eqn} with $h^{jk}$. Using \eq{sigma_Lambda} and \eq{sigma_dfn} that yields
%%
%\begin{eqnarray}
%0 &=&
%\mcD^j\mathcal{E}_{ij}
%+\frac{1}{4} \frac{q\mathcal{P}^2 -2 \Lambda}{  q\mathcal{P}^2 - \frac{\Lambda}{2}} q \mathcal{P}_i\mathcal{P}^j \mcD_j\log \mathcal{P}^2
% - q \mathcal{P}^j\mcD_j\mathcal{P}_i
%+\frac{1}{4}  \frac{q\mathcal{P}^2 + \Lambda}{q\mathcal{P}^2 - \frac{\Lambda}{2}}
%q \mcD_i\mathcal{P}^2
%\phantom{xxxx}
%\\
%&=& \mcD^j\mathcal{E}_{ij}
%+\frac{1}{4} \frac{q\mathcal{P}^2 -2 \Lambda}{  q\mathcal{P}^2 - \frac{\Lambda}{2}} q\Big( \mathcal{P}_i\mathcal{P}^j \mcD_j\log \mathcal{P}^2
%- \mcD_i\mathcal{P}^2\Big)
%+2 q \mathcal{P}^j    \mcD_{[i}\mathcal{P}_{j]}
%\;.
%\label{useful_int_res}
%\end{eqnarray}
%
Let us rewrite \eq{Equation2_contr_2}.
With \eq{cand_sigma} and  the identity
$
\mathcal{P}^j\mcD_{(i} \mathcal{P}_{j)}  \,\equiv\, \frac{1}{2}\mcD_{i} \mathcal{P}^2 - \frac{1}{2}\mathring\epsilon_{i}{}^{jk}\mathcal{P}_j\mathcal{Q}_k
%\\
%\sigma &\equiv&   \frac{3}{4} \Big(  q\mathcal{P}^2 - \frac{\Lambda}{2}\Big)^{-1}\mathcal{P}^{-2}\mathcal{P}^i \mcD_i \mathcal{P}^2
%\;,
$
we obtain, after contraction with~$\mathring\epsilon_{pq}{}^i$,
\begin{equation}
2\mathcal{P}_{[p}\mathcal{Q}_{q]}
\,=\, -\frac{1}{3} \sigma(q\mathcal{P}^2 -2 \Lambda)\mathring\epsilon_{pq}{}^k\mathcal{P}_k
+\frac{1}{4} \frac{ q\mathcal{P}^2 -2\Lambda}{  q \mathcal{P}^2-\frac{\Lambda}{2}}\mathring\epsilon_{pq}{}^k\mcD_{k} \mathcal{P}^2
-2i K_{[p}{}^k\mathcal{P}_{q]}\mathcal{P}_k
\;.
\label{anti_P_Q}
\end{equation}
If we contract this equation with $\mathcal{P}^q$ we recover \eq{Equation1B_2}.
%\tim{so anti-sym P actually not needed?}

By way of summary, the vacuum constraints, \eq{gen_vanishing_MST2}-\eq{q_PDE2} \emph{imply}
\eq{Equation1B_2} and certain components of \eq{Equation2_2}, namely its trace \eq{sigma_Lambda} and its contraction with $\mathcal{P}^j$ \eq{Equation1B_2}.

Let us bring \eq{Equation2_2} in a form which takes care of the fact that some of its components are redundant.
%do not need to be imposed.
For this purpose, set
\begin{equation}
\mathcal{A}_{ij}:=\mathcal{P}^2\mcD_{(i}\mathcal{P}_{j)}\,,
\end{equation}
 and note that \eq{sigma_Lambda} and \eq{cand_sigma} imply
\begin{equation}
\mathrm{tr}(\mathcal{A}) \,=\, -\Lambda\sigma\mathcal{P}^2\,=\,-\frac{3}{2} \Lambda \Big(  q\mathcal{P}^2 - \frac{\Lambda}{2}\Big)^{-1}\mathcal{P}^k \mathcal{P}^l\mathcal{A}_{kl}
\;.
\label{trace_A}
\end{equation}
On the other hand, we  use \eq{cand_sigma} to write \eq{Equation2_contr_2}  as
\begin{equation}
  q \mathcal{P}^4\mathring\epsilon_{i}{}^{kl}\mathcal{P}_k\mathcal{Q}_l
=
\frac{2}{3} (q\mathcal{P}^2- 2\Lambda)
\Big(\mathcal{P}^k\mathcal{A}_{ik} - \mathcal{P}^k\mathcal{P}^l \mathcal{A}_{kl}\mathcal{P}_i\Big)
- \frac{4}{3} i \Big( q \mathcal{P}^2-\frac{\Lambda}{2}\Big)\mathcal{P}^2\mathring\epsilon_{i}{}^{jk} K_{j}{}^l\mathcal{P}_l\mathcal{P}_k
\;.
\label{PQ_antisym}
\end{equation}
Finally, we employ \eq{trace_A} and \eq{PQ_antisym} to  rewrite \eq{Equation2_2},
% \ptcr{some P's should be calligraphic?}
%
\begin{eqnarray}
\mathcal{A}_{ij}
&=& \frac{3}{2}  \frac{q}{  q \mathcal{P}^2-\frac{\Lambda}{2}} \mathcal{P}_{(i}\mathcal{P}^k\mathcal{A}_{j)k}
  - \frac{1}{2} \frac{q\mathcal{P}^2 + \Lambda}{q\mathcal{P}^2 - \frac{\Lambda}{2}} \mathcal{P}^{-2}\mathcal{P}^k\mathcal{P}^l \mathcal{A}_{kl} h_{ij}
\nonumber
\\
&&\hspace{7em}
+\frac{3}{4}  \frac{q\mathcal{P}^2}{  q \mathcal{P}^2-\frac{\Lambda}{2}}\mathring\epsilon_{(i}{}^{kl} \mathcal{P}_{j)}\mathcal{P}_k \mathcal{Q}_l
+ i \mathcal{P}^2\mathring\epsilon_{(i}{}^{kl} K_{j)k}\mathcal{P}_l
\nonumber
\\
&=&2\mathcal{P}^{-2}   \mathcal{P}_{(i}\mathcal{P}^k\mathcal{A}_{j)k}
- \frac{1}{2}\mathcal{P}^{-2}\mathcal{P}^k\mathcal{P}^l \mathcal{A}_{kl} \Big( \mathcal{P}^{-2}\mathcal{P}_i\mathcal{P}_{j}
  +  h_{ij}\Big)
\nonumber
\\
&&
-\frac{\mathrm{tr}(\mathcal{A})}{2}\Big(\mathcal{P}^{-2}\mathcal{P}_i\mathcal{P}_j - h_{ij}\Big)
+  i\mathring\epsilon_{(i}{}^{kl} \Big(  \mathcal{P}^2K_{j)k}- K_{|k|}{}^m  \mathcal{P}_{j)}\mathcal{P}_m
\Big)\mathcal{P}_l
\;.
\phantom{xxxx}
\label{second_main_eqn}
\end{eqnarray}
Conversely, this equation implies \eq{Equation2_2} when using  \eq{sigma_Lambda} and  \eq{Equation2_contr_2}  (which in turn
follow from  \eq{gen_vanishing_MST2}, \eq{q_PDE2}, and the vacuum constraints).
Its trace and its contraction with $\mathcal{P}^ j$ are automatically satisfied, which reflects the fact that the same is true for the corresponding components of \eq{Equation2_2}.
Equation  \eq{second_main_eqn} therefore has two non-trivial independent components
%  merely two non-trivial components
which need to be fulfilled.

One may replace \eq{second_main_eqn} by a scalar equation. Since $(\Sigma,h_{ij})$ is a Riemannian manifold,
such an equation can immediately be obtained:
We write all terms on one side and compute its norm. Proceeding this way
% (which requires some computational effort)
 we find that
\eq{second_main_eqn}  is equivalent to
(we set $\mathcal{B}_i :=\mathcal{A}_{ij}\mathcal{P}^j$, $\mathcal{C}:=\mathcal{A}_{ij}\mathcal{P}^i\mathcal{P}^j$,
$\mathfrak{b}_i:=K_{ij}\mathcal{P}^j$, $\mathfrak{c}:=K_{ij}\mathcal{P}^i\mathcal{P}^j$),
 %\ptcr{multiply A by a power of P so that you don't have to divide by zero?}
%
\begin{eqnarray}
0 \,=\,\mathfrak{K}
&:=&
\mathcal{P}^4 |\mathcal{A}|^2
-2\mathcal{P}^2|\mathcal{B}|^2
+  \frac{1}{2}\mathcal{C}^2
  + \mathcal{P}^2\mathrm{tr}\mathcal{A}\Big( \mathcal{C}
- \frac{1}{2} \mathcal{P}^2\mathrm{tr}\mathcal{A}
\Big)
\nonumber
\\
&&
+\mathcal{P}^6\Big( 2\mathcal{P}^2|\mathfrak{b}|^2
- \mathcal{P}^4 ( |K|^2  + \frac{1}{2}K^2)
-  \frac{1}{2}\mathfrak{c}^2
-K\mathcal{P}^{2} \mathfrak{c}\Big)
\nonumber
\\
&&
+ 2 i\mathring\epsilon^{ijk}\mathcal{P}^4\mathcal{P}_k\Big(\mathcal{B}_i\mathfrak{b}_j
-\mathcal{P}^2 \mathcal{A}_{il}  K_{j}{}^{l}\Big)
\;.
\label{main_scalar}
\label{alg_cond2}
\end{eqnarray}
%

%\subsubsection{An alternative condition}
%
%\begin{eqnarray*}
% \mathcal{P}_j  &=& \frac{3}{2}( q\mathcal{P}^2)^{-1}\mathcal{P}^k\mathcal{E}_{jk}
%\\
%\mcD_i \mathcal{P}_j  &=& -\frac{3}{4} \frac{\Lambda}{q\mathcal{P}^2 - \frac{\Lambda}{2}}\mathcal{P}_j  \mcD_i \log\mathcal{P}^2
%+ ( q\mathcal{P}^2)^{-1}\mathcal{P}^k\mcD_i \mathcal{E}_{jk}
%\end{eqnarray*}
%
%(4.59), (4.58)
%\begin{eqnarray*}
%\mathcal{A}_{ij}
%&=& -\frac{3}{4} \frac{\Lambda}{q\mathcal{P}^2 - \frac{\Lambda}{2}}\mathcal{P}_{(i}  \mcD_{j)} \log\mathcal{P}^2
%+ ( q\mathcal{P}^2)^{-1}\mathcal{P}^k\mcD_{(i} \mathcal{E}_{j)k}
%\\
%&=&
% -\frac{3}{4} \frac{\Lambda}{q\mathcal{P}^2 - \frac{\Lambda}{2}}\mathcal{P}^{-2}\mathcal{P}_{(i}\mathring\epsilon_{j)}{}^{kl}\mathcal{P}_k\mathcal{Q}_l
%-\frac{3}{2} \frac{\Lambda}{q\mathcal{P}^2 - \frac{\Lambda}{2}}\mathcal{P}^{-2}\mathcal{P}^k\mathcal{P}_{(i} \mathcal{A}_{j)k}
%+ ( q\mathcal{P}^2)^{-1}\mathcal{P}^k\mcD_{(i} \mathcal{E}_{j)k}
%\\
%&=&
%- 2\Lambda\mathcal{P}^{-2}(q\mathcal{P}^2)^{-1}\mathcal{P}^k\mathcal{P}_{(i} \mathcal{A}_{j)k}
%+i\Lambda (q\mathcal{P}^2)^{-1}\mathcal{P}^{-2}\mathcal{P}_{(i} \mathring\epsilon_{j)}{}^{jk} K_{j}{}^l\mathcal{P}_l\mathcal{P}_k
%\\
%&&
% -\frac{1}{3} (q\mathcal{P}^2)^{-1}\mathcal{P}^{-2}(q\mathcal{P}^2- 2\Lambda)\mathrm{tr}(\mathcal{A})\mathcal{P}_{i}\mathcal{P}_{j}
%+ ( q\mathcal{P}^2)^{-1}\mathcal{P}^k\mcD_{(i} \mathcal{E}_{j)k}
%\end{eqnarray*}
%%
%\begin{eqnarray*}
%\mcD_{(i}\mathcal{E}_{j)k}
%&=& \mcD_{(i}(\mathring R_{j)k} + KK_{j)k} -K_{j)l}K_{k}{}^l  ) - i \mathring\epsilon_{(i}{}^{pq}\mcD_{j)}\mcD_p K_{qk}
%\end{eqnarray*}

\subsection{Cauchy data for vacuum space-times with vanishing MST}
\label{sect_Cauchy_data}

It follows from Lemma~\ref{some_intermediate} that the conditions \eq{gen_vanishing_MST2}-\eq{q_PDE2}, needed to apply
Proposition~\ref{interm_prop_2},
can be replaced by the condition
\begin{equation}
\mathcal{E}_{ij} - q_{\mathcal{C}} (\mathcal{P}_i\mathcal{P}_j)\breve{}  \,=\,0
\;,
\label{alg_cond1}
%\\
% \mathfrak{K} (\mathcal{P},\mcD\mathcal{P},h,K)  &=& 0
%\label{alg_cond2}
%\;,
\end{equation}
where
\begin{equation}
q_{\mathcal{C}}\,:=\, Q_{\mathcal{C}}|_{\Sigma} \,=\,\pm
%2^{-2/3}\varkappa
\widetilde\varkappa^{-2}(\mathcal{E}^2)^{5/6}\Big( \pm\sqrt{\mathcal{E}^2} -  \sqrt{ \frac{2}{3}}\Lambda \Big)^{-2}
\;, \quad \widetilde\varkappa\in\mathbb{C}\setminus\{0\}
\;.
\label{dfn_qC}
\end{equation}
%%
%We want to show that these conditions are sufficient, as well.
%First of all note that  \eq{alg_cond1} yields
%%
%\begin{equation}
%q_{\mathcal{C}}\,=\,\Big(\frac{3}{2}\Big)^{1/6}\widetilde\varkappa^{-1} \frac{(q_{\mathcal{C}} \mathcal{P}^2)^{5/3}}{( q_{\mathcal{C}}\mathcal{P}^2  +  \Lambda )^2}
%\;.
%\label{dfn_qC_2}
%\end{equation}
%%
One checks
%(cf.\  \cite{kerr1})
that $ \widetilde\varkappa\mapsto \lambda \widetilde\varkappa$, $\lambda\in \mathbb{C}$ implies
  $(\sigma,Y^i)\mapsto (\lambda\sigma, \lambda Y^i)$.
The constant $\widetilde\varkappa$ therefore provides a gauge freedom which reflects the freedom to choose a scale  of the KVF.
It may be set equal to $1$.

%\subsection{Vanishing of the MST $\mathcal{S}^{(\mathcal{C})}_{\alpha\beta\mu\nu}$ on $\Sigma$}
\label{section_solv_SC}

%\tim{given vacuum data $(\Sigma,h_{ij},K_{ij})$, ``MST'' \eq{unproper_MST}, exist $\mathcal{P}$...somewhere}
In this section we discuss the solvability of \eq{alg_cond1}.
Assume that \eq{all_ineqs_Cauchy1}  holds.
We deduce from \eq{F2_C2_relation},
\begin{equation}
\mathcal{P}^2 \,=\, -\frac{1}{4}\mathcal{F}^2|_{\Sigma}
%\,=\, \sqrt{\frac{3}{2}}\,\varkappa^{-1}\frac{( \sqrt{\mathcal{C}^2} -  \sqrt{ \frac{32}{3}}\Lambda )^2}{(\mathcal{C}^2)^{1/3}}
 \,=\, - \sqrt{\frac{3}{2}}\,\widetilde\varkappa^2(\mathcal{E}^2)^{-1/3}\Big(\pm\sqrt{\mathcal{E}^2} -  \sqrt{ \frac{2}{3}}\Lambda \Big)^2
\label{P_E_relation}
\;.
\end{equation}
Let us analyze the vanishing of the MST  $\mathcal{S}^{(\mathcal{C})}_{\alpha\beta\mu\nu}$ on $\Sigma$,%
\footnote{
We note  that \eq{P_E_relation} and \eq{main_condition2}   imply
$
\mathcal{P}^k \mathcal{P}^l\mathcal{E}_{kl}\mathcal{P}^{-2} = \pm \sqrt{\frac{2}{3}\mathcal{E}^2}
$,
%
%With \eq{relation_A_E_P} and \eq{P_E_relation}
and  observe that \eq{main_condition2_0}
is equivalent  \eq{main_condition}, as one should expect from Lemma~\ref{lemma_equivalence_Qs}.
}
\begin{eqnarray}
&&\mathcal{S}^{(\mathcal{C})}_{\alpha\beta\mu\nu}|_{\Sigma}\,=\, 0 \quad \Longleftrightarrow \quad \mathcal{S}^{(\mathcal{C})}_{titj}|_{\Sigma}\,=\, 0
\\
&\Longleftrightarrow &
(\mathcal{P}_i\mathcal{P}_j)\breve{}
\,=\,\mp \widetilde\varkappa^2 (\mathcal{E}^2)^{-5/6}\Big(\pm \sqrt{ \mathcal{E}^2}
-\sqrt{\frac{2}{3}}\,\Lambda\Big)^{2}\mathcal{E}_{ij}
 \label{main_condition2_0}
\\
& \Longleftrightarrow &
\mathcal{P}_i\mathcal{P}_j
\,=\, \mp  \widetilde\varkappa^2  (\mathcal{E}^2)^{-5/6}\Big(\pm \sqrt{ \mathcal{E}^2}
-\sqrt{\frac{2}{3}}\,\Lambda\Big)^{2} \Big(\mathcal{E}_{ij} \pm \sqrt{\frac{\mathcal{E}^2}{6}}\,  h_{ij}\Big)
\;.
\phantom{xxx}
\label{main_condition2}
\end{eqnarray}
%
%for some $\varkappa\in\mathbb{C}\setminus\{0\}$.

%\begin{lemma}
%\label{some_intermediate2}
%\tim{recall the defns of P and E}
%Suppose that \eq{two_ineqs} holds.
%Then a Killing initial data set $(\Sigma,h_{ij},K_{ij},\sigma,Y^i)$  yields a $\Lambda$-vacuum space-time with a KVF such that the
%associated MST vanishes for some function $Q$
%%Then the restriction  to $\Sigma$ of the MST associated to the Killing initial data set $(\Sigma,h_{ij},K_{ij},\sigma,Y^i)$ vanishes  for some function $Q$
% if and only if
%%If an only if $\mathcal{P}_i$ and $\mathcal{E}_{ij}$ are related in this way the corresponding .
%%
%\begin{equation}
%\mathcal{P}_i\mathcal{P}_j
%\,=\, \widetilde\varkappa (\mathcal{E}^2)^{-5/6}\Big( \sqrt{ \mathcal{E}^2}
%-\sqrt{\frac{2}{3}}\,\Lambda\Big)^{2} \Big(\mathcal{E}_{ij} + \sqrt{\frac{\mathcal{E}^2}{6}}\,  h_{ij}\Big)
%\quad \text{for some $\widetilde\varkappa\in\mathbb{C}\setminus\{0\}$}.
%\label{main_condition2}
%\end{equation}
%\end{lemma}

Conditions  which characterize  solvability and uniqueness of  equations of a form such as in \eq{main_condition2}, regarded as equations for $\mathcal{P}_i$,  and different approaches to
construct  solutions thereof are discussed in paper I \cite{kerr1}.
Let us summarize the results:
\begin{enumerate}
\item[(i)]A solution exists for at most one choice of $\pm$.
%Given $\widetilde\varkappa$ and $(\,\cdot\,)^{1/6}$
It is then uniquely determined
up to a sign.
%\tim{once a branch for $(\,\cdot\,)^{1/3}$ has been chosen}
\item[(ii)] A solution exists if and only if
\begin{equation}
\mathcal{E}_{ik}\mathcal{E}^k{}_j \mp \sqrt{\frac{\mathcal{E}^2}{6}} \mathcal{E}_{ij} - \frac{\mathcal{E}^2}{3}h_{ij}\,=\, 0
\;.
\label{solvability_cond}
\end{equation}
Again, this happens  at most for either $+$ or $-$.
% and this sign determines for which sign conditions such as  \eq{final_assumptions} needs to be satisfied.
%
\item[(iii)]
Let $\mathcal{W}^i$ be any vector with $|\big(\mathcal{E}_{ij} \pm \sqrt{\frac{\mathcal{E}^2}{6}}\,  h_{ij}\big)\mathcal{W}^j|^2=1$ (its existence is ensured by  $\mathcal{E}^2\ne 0$).
Then
%(for better comparison, we set as in \cite{kerr1}  $\widetilde\varkappa:= 2^{2/3}\varkappa^{-1}$)
%\tim{global sign is irrelevant}
%
\begin{equation}
\mathcal{P}_i \,=\,% \pm
 i \Big(\frac{3}{2}\Big)^{1/4}\widetilde\varkappa\,(\mathcal{E}^2)^{-1/6}\Big( \pm\sqrt{ \mathcal{E}^2}
-\sqrt{\frac{2}{3}}\,\Lambda\Big) \Big(\mathcal{E}_{ij} \pm \sqrt{\frac{\mathcal{E}^2}{6}}\,  h_{ij}\Big)\mathcal{W}^j
\label{P_via_Cauchy_data}
\end{equation}
solves  \eq{main_condition2}, supposing that a solution exists, i.e.\ supposing that \eq{solvability_cond}~holds.
\end{enumerate}

%We are therefore led to the follow intermediate result:
%%As complement to Lemma~\ref{some_intermediate} we therefore obtain:
%%\tim{combination of xxx}
%%\tim{...Corollary~\ref{cor_vanishing_MST}}
%\begin{theorem}
%\label{thm_first_main}
%Consider a vacuum Killing initial data set $(\Sigma,h_{ij},K_{ij},\sigma,Y^i)$.
%Assume that \eq{E_ineqs} holds.
%%
%Then the Cauchy development for the vacuum Einstein equations
%$R_{\mu\nu}=\Lambda g_{\mu\nu}$ has a vanishing MST which is associated  the KVF generated by $(\sigma,Y^i)$
%%in some neighborhood of the initial surface $\Sigma$
% if and only if
%\begin{equation}
%\mathcal{E}_{ik}\mathcal{E}^k{}_j - \sqrt{\frac{\mathcal{E}^2}{6}} \mathcal{E}_{ij} - \frac{\mathcal{E}^2}{3}h_{ij}\,=\, 0
%\;.
%\label{solvability_condB}
%\end{equation}
%\end{theorem}

Because of  % \eq{relation_qP_E},
\eq{solvability_cond}-\eq{P_via_Cauchy_data} we have \eq{P_E_relation} and
\begin{eqnarray}
%\mathcal{P}^2 &=& \mp\Big(\frac{3}{2}\Big)^{1/2}\widetilde\varkappa (\mathcal{E}^2)^{-1/3}\Big(\pm \sqrt{ \mathcal{E}^2}
%-\sqrt{\frac{2}{3}}\,\Lambda\Big)^2
%\;,
%\\
\mathcal{Q}_i
&=&%\pm
 i \Big(\frac{3}{2}\Big)^{1/4}\widetilde\varkappa\,(\mathcal{E}^2)^{-1/6}\mathring\epsilon_i{}^{jk}
\Big[ \frac{1}{3}\Big( \pm\sqrt{ \mathcal{E}^2}
+\frac{\Lambda}{\sqrt{6}}
\Big)
 \Big(\mathcal{E}_{kl}\mathcal{W}^l \pm \sqrt{\frac{\mathcal{E}^2}{6}}\mathcal{W}_k\Big)\mcD_j \log\mathcal{E}^2
\nonumber
\\
&&
+\Big( \pm\sqrt{ \mathcal{E}^2}
-\sqrt{\frac{2}{3}}\,\Lambda\Big) \mcD_j\Big(\mathcal{E}_{kl}\mathcal{W}^l \pm \sqrt{\frac{\mathcal{E}^2}{6}}\mathcal{W}_k\Big)
\Big]
\;.
\label{Q_via_Cauchy_data}
\end{eqnarray}
We employ  \eq{relation_qP_E} and \eq{P_via_Cauchy_data}-\eq{Q_via_Cauchy_data} to express $\sigma$ and $Y^i$
(given by \eq{cand_sigma} and \eq{candidate_field_Y})  in terms of  the Cauchy data and $\mathcal{W}^i$,
\begin{eqnarray}
 \sigma
 &=&
-\frac{1}{\sqrt{6}}
\Big(\pm\sqrt{ \mathcal{E}^2} -\sqrt{\frac{2}{3}}\,\Lambda\Big)^{-1} \mathcal{P}^k\mcD_k  \log \mathcal{E}^2
\\
 &=&
%\mp
- \frac{i}{2}\Big(\frac{2}{3}\Big)^{1/4} \widetilde\varkappa (\mathcal{E}^2)^{-1/6} \Big(\mathcal{E}^{kl}\mathcal{W}_l \pm \sqrt{\frac{\mathcal{E}^2}{6}} \mathcal{W}^k\Big) \mcD_k\log \mathcal{E}^2
\;,
\label{final_sigma}
\end{eqnarray}
and
\begin{eqnarray}
Y_i&=&
\sqrt{\frac{3}{2}} \Big( \pm\sqrt{\mathcal{E}^2} -\sqrt{\frac{2}{3}} \Lambda\Big)^{-1}\Big[ \frac{i}{2}\mathring\epsilon_i{}^{kl}
\frac{\pm\sqrt{\mathcal{E}^2}}{ \pm \sqrt{ \mathcal{E}^2}
-\sqrt{\frac{2}{3}}\,\Lambda}\mathcal{P}_k\mcD_l \log \mathcal{E}^2
\nonumber
\\
&&\qquad
+ i \mathcal{Q}_i
 +K_{i}{}^{k}\mathcal{P}_{k}-  K \mathcal{P}_{i}
\Big]
\\
&=&
%\pm
 \Big(\frac{3}{2}\Big)^{3/4}\widetilde\varkappa(\mathcal{E}^2)^{-1/6}
 \Big[-\mathring\epsilon_i{}^{jk}\mcD_j
 +\frac{1}{6}\mathring\epsilon_i{}^{jk}\mcD_j \log\mathcal{E}^2
 +i K_{i}{}^{k}- iK\delta_i{}^k
\Big]
\nonumber
\\
&&\qquad
\times \Big(\mathcal{E}_{k}{}^{l} \mathcal{W}_l\pm \sqrt{\frac{\mathcal{E}^2}{6}}\mathcal{W}_k\Big)
\;.
\label{final_Y}
\end{eqnarray}

%In contrast,  the freedom to choose  $(\mathcal{E}^2)^{1/6}$ leads to an ambiguity in $ \widetilde\varkappa$, which is determined  merely  up to
%multiplication with  3rd roots of unity. For this reason $\mathcal{P}_i$ is determined up to 4th roots of unity.
%In spite of this ambiguity, $\sigma$ and $Y^i$, which depend only upon the combination $(\widetilde\kappa^{-3} \mathcal{E}^2)^{1/6}$,  are uniquely determined by \eq{final_sigma} and \eq{final_Y}.

%\emph{For given KIDs},   \eq{main_condition2}, or \eq{solvability_cond}, \emph{characterizes} the vanishing of the MST in the emerging space-time.
%It therefore may replace \eq{main_condition} in Lemma~\ref{some_intermediate} (we assume that   \eq{E_ineqs} holds).
%Once we know that the space-time admits a KVF such that the associated MST vanishes on $\Sigma$ for $Q=Q_{\mathcal{C}}$, we may apply
%Lemma~\ref{lemma_equivalence_Qs} and Corollary~\ref{cor_vanishing_MST} to conclude that  $\mathcal{S}^{(\mathrm{ev})}_{\alpha\beta\mu\nu}$ vanishes in some neighborhood of $\Sigma$.

\begin{remark}
{\rm
As for \eq{second_main_eqn}, since $(\Sigma,h_{ij})$ is a Riemannian manifold, \eq{solvability_cond} can be replaced by the equation which requires the vanishing of the norm of its
left-hand side,
\begin{equation}
\mathfrak{H}\,:=\,\mathrm{tr}(\mathcal{E}\cdot\mathcal{E}\cdot\mathcal{E}\cdot\mathcal{E})
 \mp \sqrt{\frac{2}{3}}\sqrt{\mathrm{tr}(\mathcal{E}\cdot\mathcal{E})}\,\mathrm{tr}(\mathcal{E}\cdot\mathcal{E}\cdot\mathcal{E})
- \frac{1}{6}(\mathrm{tr}(\mathcal{E}\cdot\mathcal{E}))^2
\,=\, 0
\;.
\label{solvability_condB_norm}
\end{equation}
}
\end{remark}

We have proven the first main result, which provides an algorithmic characterization of Cauchy data which generate vacuum space-times with vanishing MST
(cf.\ \cite[Theorem~4.8]{kerr1} for the space-time pendant).
It  brings together all the results of the previous sections.
\begin{theorem}
\label{thm_first_main_result}
Consider Cauchy data $(\Sigma,h_{ij}, K_{ij})$ which solve the vacuum constraint equations and satisfy
(cf.\ footnote~\ref{footnote_E_conditions})
\begin{equation}
 \mathrm{tr}(\mathcal{E}\cdot\mathcal{E}) \,\ne \,0
\;, \enspace
\mathrm{tr}(\mathcal{E}\cdot\mathcal{E})   -\frac{2}{3}\,\Lambda^2 \,\ne \, 0
\;, \enspace
\mathrm{tr}(\mathcal{E}\cdot\mathcal{E})   - \frac{1}{6}\,\Lambda^2 \,\ne \, 0
\;, \enspace
\mathrm{tr}(\mathcal{E}\cdot\mathcal{E})   - \frac{8}{3}\,\Lambda^2 \,\ne \, 0
\;,
\label{final_assumptions}
\end{equation}
where
\begin{equation*}
\mathcal{E}_{ij} \,:=\, \mathring R_{ij}  + K K_{ij} - K_{ik}K_j{}^k -\frac{2}{3}\Lambda h_{ij}
 -i\mathring\epsilon_{i}{}^{kl}\mcD_{k}K_{lj}
%\;, \quad \mathcal{E}^2\,:=\, \mathcal{E}_{ij}\mathcal{E}^{ij}
\;.
\end{equation*}
Moreover, let $\mathcal{W}^i$ be any vector with $|\big(\mathcal{E}_{ij} \pm \sqrt{\frac{\mathcal{E}^2}{6}}\,  h_{ij}\big)\mathcal{W}^j|^2=1$
 (which exists),
and set
%\tim{add some remark concering $\pm$}
%
\begin{equation}
\mathcal{P}_i \,:=\,  %\pm
 i \Big(\frac{3}{2}\Big)^{1/4}\widetilde\varkappa\,(\mathcal{E}^2)^{-1/6}\Big( \pm\sqrt{ \mathcal{E}^2}
-\sqrt{\frac{2}{3}}\,\Lambda\Big) \Big(\mathcal{E}_{ij} \pm \sqrt{\frac{\mathcal{E}^2}{6}}\,  h_{ij}\Big)\mathcal{W}^j
\label{P_via_Cauchy_data2}
\end{equation}
(the ``right'' signs are determined by condition (i) below).
Then the emerging $\Lambda$-vacuum space-time admits a non-trivial (possibly complex) KVF such that the associated MST vanishes
(at least in some neighborhood of $\Sigma$)
if and only if  %\eq{solvability_condB}
 \eq{alg_cond2} and \eq{solvability_condB_norm} hold, i.e.
\begin{enumerate}
\item[(i)]
%$\mathcal{E}_{ik}\mathcal{E}^k{}_j - \sqrt{\frac{\mathcal{E}^2}{6}} \mathcal{E}_{ij} - \frac{\mathcal{E}^2}{3}h_{ij}\,=\, 0$, and
$\mathfrak{H}(\mathcal{E},h) \equiv \mathrm{tr}(\mathcal{E}\cdot\mathcal{E}\cdot\mathcal{E}\cdot\mathcal{E})
 \mp \sqrt{\frac{2}{3}\mathrm{tr}(\mathcal{E}\cdot\mathcal{E})}\,\mathrm{tr}(\mathcal{E}\cdot\mathcal{E}\cdot\mathcal{E})
- \frac{1}{6}[\mathrm{tr}(\mathcal{E}\cdot\mathcal{E})]^2
\,=\, 0$, and
\item[(ii)]
 $\mathfrak{K} (\mathcal{E},\mcD\mathcal{E},h,K)  \,=\,0$.
% (cf.\ \eq{main_scalar}).
\end{enumerate}
%(The co-vector field $\mathcal{P}$ needed to check (ii) can be computed whenever (i) holds, cf.\  Section~\ref{section_solv_SC}.)
In that case the Cauchy data are complemented to Killing initial data via \eq{final_sigma} and \eq{final_Y}.
\end{theorem}

\begin{remark}
{\rm
(i) and (ii) may be replaced by their tensor-equivalents \eq{solvability_cond} and  \eq{second_main_eqn}.
Alternatively, they may be combined into one single scalar equation,
\begin{equation}
\mathfrak{L} \,:=\, \mathfrak{H}^2 + \mathfrak{K}^2 \,=\, 0
\;,
\end{equation}
which depends only on $h_{ij}$, $K_{ij}$ and derivatives thereof.
}
\end{remark}

\begin{corollary}
The function $\mathfrak{L}$ provides, under the hypotheses \eq{final_assumptions}, a measure for the  deviation of  $\Lambda$-vacuum initial data $(\Sigma, h_{ij}, K_{ij})$
to initial data  which admit a (possibly complex) KVF whose  associated MST vanishes.
\end{corollary}

\begin{remark}
{\rm
\label{remark_reality}
The conditions \eq{final_assumptions} only make sure that the evolution equations for the MST are regular \emph{near $\Sigma$}, whence the vanishing of the MST can merely be concluded in a corresponding neighborhood.
If the KIDs are real, though,  the KVF will be real as well, and the KIDs will generate one of the vacuum space-times
 contained in the class of space-times described in \cite{mars_senovilla}.
All these space-times have the property that the MST vanishes everywhere,
%is always identically zero.
whence we conclude  that the MST actually vanishes
in the whole domain of dependence of $\Sigma$,  and not just in some neighborhood of $\Sigma$. In fact, one should expect that the same is true for MSTs associated to complex KVFs.
}
\end{remark}

\section{Algorithmic characterization of Cauchy data for the Kerr-NUT-(A)dS family}
\label{section_4}

\subsection{Vanishing of the MST associated to real KIDs and the Kerr-NUT-(A)dS family}

A necessary condition for Cauchy data $(\Sigma, h_{ij}, K_{ij})$ to generate a member of the Kerr-NUT-(A)dS family is
that the MST vanishes w.r.t.\ a KVF which is \emph{real}.
This will be the case whenever there  exists a choice $\widetilde\varkappa\in\mathbb{C}\setminus\{0\}$ for which the Killing initial data $\sigma$ and $Y^i$, as given by
\eq{final_sigma} and \eq{final_Y}, are real.
In that case there only remains the freedom to multiply $\widetilde\varkappa$ with real constants $\lambda\in\mathbb{R}\setminus\{0\}$.

Once it is known that the initial data set $(\Sigma,h_{ij},K_{ij})$ leads to a $\Lambda$-vacuum space-time which admits a
\emph{real} KVF w.r.t.\ which the MST vanishes,
the characterization result in \cite{mars_senovilla} (which we have recalled in paper I \cite{kerr1}) can be consulted to check whether the emerging space-time belongs to the Kerr-NUT-(A)dS family.
Moreover, that result can be used to compute the Kerr-NUT-(A)dS parameters $m$, $a$ and $\ell$ from $h_{ij} $ and $K_{ij}$, so  that one gains
insight which member of the Kerr-NUT-(A)dS family is generated by $(\Sigma,h_{ij},K_{ij})$.
For this we need to determine the constants $b_1$, $b_2$, $c$ and $k$ \eq{equation_b1b2}-\eq{equation_k}.

Using \eq{relation_qP_E},  \eq{dfn_qC}, and \eq{P_E_relation}
%\tim{more/less sign ambiguities?}
%\tim{6th roots, same for paper I}
%%
%\begin{equation}
%\mathcal{F}^2|_{\Sigma}\,=\, -4\mathcal{P}^2 \;, \quad   q \mathcal{P}^2\,=\, \mp\sqrt{\frac{3}{2}} \sqrt{\mathcal{E}^2}
%\;, \quad
%q\,=\,q_{\mathcal{C}}\,=\,  \pm\widetilde\varkappa^{-1} \frac{( \mathcal{E}^2 )^{5/6}}{(  \pm\sqrt{\mathcal{E}^2}  -\sqrt{\frac{2}{3}}  \Lambda )^2}
%\;,
%\end{equation}
%
we obtain (cf.\ \cite{kerr1}, but note that $\widetilde\varkappa$ differs from the one used there)
\begin{eqnarray}
 b_1
&=& % \pm
 18  \Big(  \frac{2}{3}\Big)^{1/4}  \mathrm{Im}(  \widetilde\varkappa^{3} )
\label{equation_b1_expl}
\,,
\\
 b_2
&=& %\mp
-  18 \Big(  \frac{2}{3}\Big)^{1/4} \mathrm{Re} ( \widetilde\varkappa^{3} )
\label{equation_b2_expl}
\,,
\\
c
 &=& \sigma^2 - |Y|^2\mp 6\mathrm{Re}\Big( \widetilde\varkappa^2  (\mathcal{E}^2)^{1/6}\Big) -  \sqrt{6}\,\Lambda\,\mathrm{Re}\Big( \widetilde\varkappa^2   (\mathcal{E}^2)^{-1/3} \Big)
\,,
\label{equation_c_expl}
\\
k
 &=& 9  \Big(\frac{2}{3}\Big)^{1/2} |\widetilde\varkappa^2( \mathcal{E}^2 )^{-1/3} | \Big(\mcD_{i}Z\mcD^{i}Z - (\ol{\nabla_{0}Z})^2 \Big) -  b_2Z +cZ^2 +\frac{\Lambda}{3} Z^4
\,,
\phantom{xx}
\label{equation_k_expl}
\end{eqnarray}
where
\begin{eqnarray}
 Z|_{\Sigma}
&=&3\Big(\frac{2}{3}\Big)^{1/4}\mathrm{Re} \Big( \widetilde\varkappa( \mathcal{E}^2 )^{-1/6}\Big)
\,,
\label{equation_Z_expl}
\\
\nabla_{0}Z |_{\Sigma}
&\overset{\cite{kerr1}}{=}& -\frac{3}{2}\,\mathrm{Re}\Big(  \frac{ \nabla_{0}\mathcal{F}^2}{\sqrt{\mathcal{F}^2}(Q\mathcal{F}^2 +2\Lambda)}
\Big)\Big|_{\Sigma}
\\
 &\overset{\eq{nabla_F2}}{=}& \Big(\frac{2}{3}\Big)^{1/4}\mathrm{Re}\Big(  \widetilde\varkappa^{-1} \frac{( \mathcal{E}^2 )^{1/6}}{\pm \sqrt{\mathcal{E}^2}  -\sqrt{\frac{2}{3}}  \Lambda } Y^{i}\mathcal{P}_{i }
\Big)
\;.
\label{gen_nablaZ}
\end{eqnarray}
\begin{remark}
\label{remark_weaker}
{\rm
The 6th root $(\mathcal{E}^2)^{1/6}$ is determined by the requirement that the Killing initial data $(\sigma, Y)$ need to be real.
}
\end{remark}

Since it is of particular physical interest and somewhat easier to analyze, we devote ourselves henceforth to Kerr-(A)dS family.

\subsection{Kerr-(A)dS family}

To end up with  an algorithmic local characterization result for the Kerr-(A)dS metrics  in terms of Cauchy data we will
employ the space-time characterization Theorem~\ref{thm_charact}.
%\tim{!!!!!!! and afterwards}
We assume that we have been given Cauchy data $(\Sigma, h_{ij},K_{ij})$  which fulfill  all hypotheses of Theorem~\ref{thm_first_main_result}.
%Recall that $\mathcal{F}^2|_{\Sigma} =-4\mathcal{P}^2$ and $Q_{\mathcal{C}}|_{\Sigma}=q_{\mathcal{C}}$  whence
In particular   \eq{final_assumptions} implies that $Q\mathcal{F}^2$ and $Q\mathcal{F}^2-4\Lambda$
are not identically zero, as required by Theorem~\ref{thm_charact}.
% which then also holds in some neighborhood of $\Sigma$.
Then we supplement the data via \eq{cand_sigma} and \eq{candidate_field_Y} to Killing initial data  $(\Sigma, h_{ij},K_{ij}, \sigma, Y^i)$,
where we assume that there exists a choice of $\widetilde\varkappa\in\mathbb{C}\setminus\{0\}$ for which  $\sigma$ and $Y^i$ are real (in other words we  require
 $(\sigma, Y^i)$  to be real up to some multiplicative complex constant).

According to Theorem~\ref{thm_charact}  a necessary condition  for the Cauchy development of $(\Sigma,h_{ij}, K_{ij})$ to be locally isometric
to a Kerr-(A)dS space-time is $b_2 = 0$, equivalently   $\mathrm{Re} ( \widetilde\varkappa^{3} )=0$.
There remains a gauge freedom
concerning the choice of the  complex constant  $\widetilde\varkappa$, namely  to prescribe
its  length.  It arises from the freedom to rescale the KVF.
%We can further employ the freedom to choose a 6th root of unity:
One may therefore  impose the gauge condition
\begin{equation}
\widetilde\varkappa %^3
 \,=\, i
\,.
\label{gauge_tildekappa}
\end{equation}
%
%(and assume that  $(\sigma,Y^i)$ is real).
%Then \eq{equation_b2_expl} becomes
%%
%\begin{eqnarray}
% b_1
%&=& \mp  18  \Big(  \frac{2}{3}\Big)^{1/4}
%\label{equation_b1_expl_s}
%\,.
%\end{eqnarray}

Let us analyze the validity of \eq{final_assumptions} in the KdS-case.
In \cite{kerr1} it has been shown that the Kerr-(A)dS family  satisfies
\begin{equation}
\mathcal{C}^2 \,=\, \frac{96m^2}{(r+ia\cos\theta)^6}
\;,
\end{equation}
and as in \cite{kerr1} we define $\sqrt{\,\cdot\,}$ in such a way that
\begin{equation}
\sqrt{\mathcal{C}^2} \,=\, \frac{\sqrt{96} m}{(r+ia\cos\theta)^3}
\end{equation}
(then \eq{solvability_condB_norm} holds with ``$-$'').
For $m\ne0$ we thus have  $\mathcal{E}^2\ne 0$ on any Cauchy surface $\Sigma$.
In particular  \eq{final_assumptions}  holds everywhere in the $\Lambda\ne0$-case.
Moreover, observe that that $\mathrm{grad}(\mathrm{Re}[(\mathcal{C}^2)^{-1/6}])$ is nowhere vanishing.

So let us consider the case $\Lambda\ne 0$ (and $m\ne 0 $).
It has been shown  in \cite{kerr1} that  $\sqrt{\mathcal{C}^2}\ne \sqrt{\frac{32}{3}}\Lambda$  holds if and only if
\begin{eqnarray}
\text{for $a=0$:} &&    r  \,=\,  \big(3 m\Lambda^{-1}\big)^{1/3}
\;,
\label{possi1}
\\
\text{for $a>  0$:} && \theta  \,=\,\pi/2
\quad \text{and} \quad
    r  \,=\,  \big(3 m\Lambda^{-1}\big)^{1/3}
\label{possi2}
\\
&& \hspace{5em}  \text{or}
\nonumber
\\
&&
  \cos\theta \,=\, \pm \Big(\frac{9}{8}\sqrt{3}\, ma^{-3} \Lambda^{-1}\Big)^{1/3}
\quad \text{and} \quad
r   \,=\,  \mp \frac{a}{\sqrt{3}}  \cos\theta
\;.
\phantom{xx}
\label{possi3}
\end{eqnarray}
Clearly, the solution \eq{possi3} exists only for $\frac{9}{8}\sqrt{3}\, ma^{-3} \Lambda^{-1}\leq 1$.
Moreover,
$\sqrt{\mathcal{C}^2}\ne -\sqrt{\frac{8}{3}}\Lambda$  is equivalent to \cite{kerr1}
\begin{eqnarray}
\text{for $a=0$:} &&    r  \,=\, - \big(6 m\Lambda^{-1}\big)^{1/3}
\;,
\label{possiB1}
\\
\text{for $a>  0$:} && \theta  \,=\,\pi/2
\quad \text{and} \quad
    r  \,=\,  -\big(6 m\Lambda^{-1}\big)^{1/3}
\label{possiB2}
\\
&& \hspace{5em}  \text{or}
\nonumber
\\
&&
  \cos\theta \,=\, \pm \Big(\frac{9}{4}\sqrt{3}\, ma^{-3} \Lambda^{-1}\Big)^{1/3}
\quad \text{and} \quad
 r   \,=\,  \pm \frac{a}{\sqrt{3}}  \cos\theta
\;, \phantom{xxx}
\label{possiB3}
\end{eqnarray}
where \eq{possiB3} exists only for $\frac{9}{4}\sqrt{3}\, ma^{-3} \Lambda^{-1}\leq 1$.
It remains to consider the last condition in \eq{final_assumptions}:
%Using that for K(A)dS we have
%
%\begin{equation}
%\sqrt{\mathcal{C}^2} \,=\, \frac{\sqrt{96} \,m}{ (r + i a\cos\theta)^3}
%\;,
%\end{equation}
%%
%we conclude that
%
\begin{eqnarray}
\sqrt{\mathcal{C}^2} \,=\,- \sqrt{ \frac{128}{3}}\Lambda     & \Longleftrightarrow  &
  \frac{3}{2}m \Lambda^{-1}   \,=\, - (r + i a\cos\theta)^3
\;,
\label{ineq_eq3}
\end{eqnarray}
happens if and only if
\begin{eqnarray}
\text{for $a=0$:} &&    r  \,=\, - \Big(\frac{3}{2} m\Lambda^{-1}\Big)^{1/3}
\;,
\label{possiC1}
\\
\text{for $a>  0$:} && \theta  \,=\,\pi/2
\quad \text{and} \quad
    r  \,=\,  -\Big(\frac{3}{2} m\Lambda^{-1}\Big)^{1/3}
\label{possiC2}
\\
&& \hspace{5em}  \text{or}
\nonumber
\\
&&
  \cos\theta \,=\, \pm \Big(\frac{9}{16}\sqrt{3}\, ma^{-3} \Lambda^{-1}\Big)^{1/3}
\quad \text{and} \quad
 r   \,=\,  \pm \frac{a}{\sqrt{3}}  \cos\theta
\;. \phantom{xxx}
\label{possiC3}
\end{eqnarray}
The solution \eq{possiC3} exists only for $\frac{9}{16}\sqrt{3}\, ma^{-3} \Lambda^{-1}\leq 1$.

To sum it up,   for $\Lambda=0$, the conditions \eq{final_assumptions}  are satisfied everywhere by the Kerr family.
For $\Lambda\ne 0 $ and $a\ne 0$ some of  the conditions in \eq{final_assumptions}  are violated on certain   $\{r,\theta=\mathrm{const.}\}$-2-surfaces,
in the Schwarzschild-(A)de Sitter case   on certain
$\{r=\mathrm{const.}\}$-hypersurfaces.
For the latter ones,
\begin{equation}
g(\partial_r,\partial_r) \,=\, \Big(1  - \frac{2m}{r} - \frac{\Lambda}{3}r^2\Big)^{-1}
\;,
\end{equation}
so for appropriate choices of $\Lambda$ and $m$ these surfaces will be spacelike.
Our results do not apply on these 3-surfaces.
%\tim{!!}
%here. We need to exclude them.

Finally, we state our second main result:

\begin{theorem}
\label{thm_second_main_result}
Consider Cauchy data $(\Sigma,h_{ij}, K_{ij})$ which solve the vacuum constraint equations and satisfy
(cf.\ footnote~\ref{footnote_E_conditions})
%\tim{not needed in the assumptions...?}
%\tim{except possibly on co-dimension $\geq 1$ surfaces, add footnote}
%
\begin{equation*}
 \mathrm{tr}(\mathcal{E}\cdot\mathcal{E}) \,\ne \,0
\;, \enspace
\mathrm{tr}(\mathcal{E}\cdot\mathcal{E})   - \frac{2}{3}\Lambda^2 \,\ne \, 0
\;, \enspace
\mathrm{tr}(\mathcal{E}\cdot\mathcal{E}) - \frac{1}{6}\Lambda^2 \,\ne \, 0
\;, \enspace
\mathrm{tr}(\mathcal{E}\cdot\mathcal{E})   - \frac{8}{3}\Lambda^2 \,\ne \, 0
\;,
\end{equation*}
where
\begin{equation*}
\mathcal{E}_{ij} \,:=\, \mathring R_{ij}  + K K_{ij} - K_{ik}K_j{}^k -\frac{2}{3}\Lambda h_{ij}
 -i\mathring\epsilon_{i}{}^{kl}\mcD_{k}K_{lj}
%\;, \quad \mathcal{E}^2\,:=\, \mathcal{E}_{ij}\mathcal{E}^{ij}
\;,
\end{equation*}
and for which $\mathrm{Im}\Big(\frac{\sqrt{\mathcal{F}^2}}{Q\mathcal{F}^2-4\Lambda}\Big)$ has non-zero gradient somewhere.

Then the emerging $\Lambda$-vacuum space-time
is locally isometric to a member of the Kerr-(A)dS family if and only if (i)-(iv) hold:
\begin{enumerate}
\item[(i)]
%$\mathcal{E}_{ik}\mathcal{E}^k{}_j - \sqrt{\frac{\mathcal{E}^2}{6}} \mathcal{E}_{ij} - \frac{\mathcal{E}^2}{3}h_{ij}\,=\, 0$, and
$\mathrm{tr}(\mathcal{E}\cdot\mathcal{E}\cdot\mathcal{E}\cdot\mathcal{E})
 \mp\sqrt{\frac{2}{3}\mathrm{tr}(\mathcal{E}\cdot\mathcal{E})}\,\mathrm{tr}(\mathcal{E}\cdot\mathcal{E}\cdot\mathcal{E})
- \frac{1}{6}[\mathrm{tr}(\mathcal{E}\cdot\mathcal{E})]^2
\,=\, 0$.
\end{enumerate}
 Let $\mathcal{W}^i$ be any
%\tim{real}
vector field with $|\big(\mathcal{E}_{ij} \pm \sqrt{\frac{\mathcal{E}^2}{6}}\,  h_{ij}\big)\mathcal{W}^j|^2=1$
 (which exists because $\mathrm{tr}(\mathcal{E}\cdot\mathcal{E}) \ne 0$ );
then,
%for  $\widetilde\varkappa\in\mathbb{C}\setminus\{0\}$,
set
\begin{eqnarray*}
\mathcal{P}_i &=&
 - \Big(\frac{3}{2}\Big)^{1/4}%\widetilde\varkappa\,
(\mathcal{E}^2)^{-1/6}\Big( \pm\sqrt{ \mathcal{E}^2}
-\sqrt{\frac{2}{3}}\,\Lambda\Big) \Big(\mathcal{E}_{ij} \pm \sqrt{\frac{\mathcal{E}^2}{6}}\,  h_{ij}\Big)\mathcal{W}^j
\;,
%\\
%q &=&\mp
%%\widetilde\varkappa^{-2}
%(\mathcal{E}^2)^{-1/6}\mathcal{E}^2\Big( \pm\sqrt{\mathcal{E}^2} -  \sqrt{ \frac{2}{3}}\Lambda \Big)^{-2}
\end{eqnarray*}
(the signs are determined by (i)).
\begin{enumerate}
\item[(ii)]
 $\mathfrak{K}  \,=\,0$ (the scalar $\mathfrak{K}$ has been defined in \eq{main_scalar} in terms of the Cauchy data and $\mathcal{P})$),
\item[(iii)] The fields  $\sigma$ and $Y$ are real, where
\begin{eqnarray*}
\sigma &=&%\mp
\frac{1}{2}\Big(\frac{2}{3}\Big)^{1/4}
%\widetilde\varkappa
(\mathcal{E}^2)^{-1/6} \Big(\mathcal{E}^{kl}\mathcal{W}_l \pm \sqrt{\frac{\mathcal{E}^2}{6}} \mathcal{W}^k\Big) \mcD_k\log \mathcal{E}^2
\;,
\\
Y^i &=&
%\pm
i\Big(\frac{3}{2}\Big)^{3/4}
%\widetilde\varkappa
(\mathcal{E}^2)^{-1/6}
 \Big[-\mathring\epsilon^{ijk}\mcD_j
 +\frac{1}{6}\mathring\epsilon^{ijk}\mcD_j \log\mathcal{E}^2
 +i K^{ik}- iKh^{ik}
\Big]
\\
&&\qquad \times  \Big(\mathcal{E}_{k}{}^{l} \mathcal{W}_l\pm \sqrt{\frac{\mathcal{E}^2}{6}}\mathcal{W}_k\Big)
\;.
\end{eqnarray*}
\item[(iv)] %$\mathrm{grad}(\mathrm{Im} \Big( \widetilde\varkappa( \mathcal{E}^2 )^{-1/6}\Big))$
$\mathrm{grad}(\mathrm{Re} [( \mathcal{E}^2 )^{-1/6}])$
 is not identically zero, and
%\tim{add comment??}
\item[(v)] the constants $c$ and $k$, given by \eq{equation_c_expl}-\eq{equation_k_expl} (with \eq{gauge_tildekappa}), satisfy, depending on the sign of the cosmological constant, \eq{second_condition1}-\eq{second_condition3}, respectively.
\end{enumerate}
%(The co-vector field $\mathcal{P}$ needed to ckeck (ii) can be computed whenever (i) holds, cf.\  Section~\ref{section_solv_SC}.)
%In that case the Cauchy data are complemented to Killing initial data via \eq{cand_sigma} and \eq{candidate_field_Y}.
If (i)-(iv) are fulfilled, the K(A)dS-space-time generated by $(\Sigma,h_{ij}, K_{ij})$ has parameters
%\tim{sign ambiguity}
%
\begin{equation}
m\,=\,
%\pm
9  \Big(  \frac{2}{3}\Big)^{1/4}
% \mathrm{Im}(  \widetilde\varkappa^{3} )
\big(\frac{\Lambda}{3} \zeta_1^2 + c\big)^{-3/2}\,, \quad a\,=\, \zeta_1\big(\frac{\Lambda}{3} \zeta_1^2 + c\big)^{-1/2}\,,
\end{equation}
where $\zeta_1$ is given by \eq{first_condition1b}-\eq{first_condition3b}.
% and $c$ and $k$ by  \eq{equation_c_expl_s}-\eq{gen_nablaZ_s}.
The KVF whose associated MST vanishes (in Boyer-Lindquist-type coordinates this is a multiple of the $\partial_t$-KVF restricted to $\Sigma$)
is then generated by $(\sigma,Y)$.
\end{theorem}

An  issue of interest would be to do an analog analysis for the characteristic initial value problem.

\vspace{1.2em}
\noindent {\textbf {Acknowledgements}}
I am grateful to  Piotr  Chru\'sciel  for many helpful comments to improve the manuscript.
The research  was funded by the Austrian Science Fund (FWF): P 23719-N16 and P 28495-N27.

\end{document}

%% file: LeeNewcommand.tex
\newcounter{mnotecount}[section]

\renewcommand{\themnotecount}{\thesection.\arabic{mnotecount}}
\newcommand{\mnote}[1]%{}
{\protect{\stepcounter{mnotecount}}$^{\mbox{\footnotesize
$%\!\!\!\!\!\!\,
\bullet$\themnotecount}}$ \marginpar{%\color{red}%
\raggedright\tiny\em
$\!\!\!\!\!\!\,\bullet$\themnotecount: #1} }

\newcommand{\tim}[1]{\mnote{{\bf tim:}#1}}

\newtheorem{theorem}{\sc  Theorem\rm}[section]

\newtheorem{corollary}[theorem]{\sc  Corollary\rm}

\newtheorem{lemma}[theorem]{\sc Lemma\rm}

\newtheorem{proposition}[theorem]{\sc Proposition\rm}

\newtheorem{remark}[theorem]{\sc Remark\rm}

\newcommand{\ol}[1]{\overline{#1}{}}

\newcommand{\jlcax}[1]{}
\newcommand{\eean}{\nonumber\end{eqnarray}}

\newcommand{\kk}[1]{}%{\mnote{{\bf If we consider the KK case:} #1}}

\newcommand{\beq}{\begin{equation}}

\newcommand{\FS}       %{F_1} %
                  {F}
\newcommand{\HS} %{F_2}
       {H_{\mbox{\scriptsize volume}}}

{\ptc{this should be removed in the oberwolfach version}}%

\newcommand{\eeal}[1]{\label{#1}\end{eqnarray}}
\newcommand{\bed}{\begin{deqarr}}
\newcommand{\eed}{\end{deqarr}}
\newcommand{\bedl}[1]{\begin{deqarr}\label{#1}}
\newcommand{\eedl}[2]{\arrlabel{#1}\label{#2}\end{deqarr}}

\newcommand{\bel}[1]{\begin{equation}\label{#1}}
\newcommand{\bea}{\begin{eqnarray}}
\newcommand{\bean}{\begin{eqnarray}\nonumber}
\newcommand{\beal}[1]{\begin{eqnarray}\label{#1}}
\newcommand{\eea}{\end{eqnarray}}

\newcommand{\qed}{\hfill $\Box$ \medskip}

\newcommand{\be}{\begin{equation}}
\newcommand{\eeq}{\end{equation}}
\newcommand{\ee}{\end{equation}}
\newcommand{\beqa}{\begin{eqnarray}}
\newcommand{\eeqa}{\end{eqnarray}}
\newcommand{\beqan}{\begin{eqnarray*}}
\newcommand{\eeqan}{\end{eqnarray*}}
\newcommand{\ba}{\begin{array}}
\newcommand{\ea}{\end{array}}

\newcommand{\mcM}{{\mycal M}}
\newcommand{\mcD}{{\mycal D}}

\newcommand{\scri}{{\mycal I}}%

\newcommand{\warn}[1]%{}%{}
{\protect{\stepcounter{mnotecount}}$^{\mbox{\footnotesize
$%\!\!\!\!\!\!\,
\bullet$\themnotecount}}$ \marginpar{%\color{red}%
\raggedright\tiny\em
$\!\!\!\!\!\!\,\bullet$\themnotecount: {\bf Warning:} #1} }

\newcommand{\eq}[1]{(\ref{#1})}

\newcommand{\ptc}[1]{\mnote{{\bf ptc:}#1}}

\newcommand{\mcL}{{\mycal L}}

\newcommand{\beqar}{\begin{deqarr}}
\newcommand{\eeqar}{\end{deqarr}}

\newcommand{\beaa}{\begin{eqnarray*}}
\newcommand{\eeaa}{\end{eqnarray*}}